\documentclass[11pt]{article}

\usepackage{amsmath}
\usepackage{amssymb}
\usepackage{amsfonts,amsthm}
\usepackage{thmtools,thm-restate}
\usepackage{enumitem}
\usepackage{array}
\usepackage{multicol}

\newtheorem{theorem}{Theorem}[section]
\newtheorem{lemma}[theorem]{Lemma}

\newtheorem{observation}[theorem]{Observation}
\theoremstyle{definition}
\newtheorem{definition}[theorem]{Definition}

\usepackage{tikz}
\usepackage[a4paper,bindingoffset=0.2in,
            top=1.1in, bottom=1.1in, left=1.0in, right=1.0in,
            footskip=.5in]{geometry}

\newcommand{\df}{_{\rm def}}
\newcommand{\lmin}{\ell\text{-}\min}
\newcommand{\rmax}{r\text{-}\max}

\newcommand{\myspace}{2pt}

\newcommand{\ignore}[1]{}

\begin{document}

\title{Polynomial-time Algorithms for the Subset Feedback Vertex Set Problem on Interval Graphs and Permutation Graphs}

\author{
Charis Papadopoulos\thanks{Department of Mathematics, University of Ioannina, Greece. E-mail:  \texttt{charis@cs.uoi.gr}}
\and
Spyridon Tzimas\thanks{Department of Mathematics, University of Ioannina, Greece. E-mail: \texttt{roytzimas@hotmail.com}}
}

\date{}
\pagestyle{plain}
\maketitle

\begin{abstract}
Given a vertex-weighted graph $G=(V,E)$ and a set $S \subseteq V$, a subset feedback vertex set $X$ is a set of the vertices of $G$ such that the graph induced by $V \setminus X$ has no cycle containing a vertex of $S$.
The \textsc{Subset Feedback Vertex Set} problem takes as input $G$ and $S$ and asks for the subset feedback vertex set of minimum total weight.
In contrast to the classical \textsc{Feedback Vertex Set} problem which is obtained from the \textsc{Subset Feedback Vertex Set} problem for $S=V$,
restricted to graph classes the \textsc{Subset Feedback Vertex Set} problem is known to be NP-complete on split graphs and, consequently, on chordal graphs.
However as \textsc{Feedback Vertex Set} is polynomially solvable for AT-free graphs, no such result is known for the \textsc{Subset Feedback Vertex Set} problem on any subclass of AT-free graphs.
Here we give the first polynomial-time algorithms for the problem on two unrelated subclasses of AT-free graphs: interval graphs and permutation graphs.
As a byproduct we show that there exists a polynomial-time algorithm for circular-arc graphs by suitably applying our algorithm for interval graphs.
Moreover towards the unknown complexity of the problem for AT-free graphs, we give a polynomial-time algorithm for co-bipartite graphs.
Thus we contribute to the first positive results of the \textsc{Subset Feedback Vertex Set} problem when restricted to graph classes for which \textsc{Feedback Vertex Set} is solved in polynomial time.
\end{abstract}

\section{Introduction}
For a given set $S$ of vertices of a graph $G$, a {\it subset feedback vertex set} $X$ is a set of vertices such that every cycle of $G[V \setminus X]$ does not contain a vertex from $S$.
The \textsc{Subset Feedback Vertex Set} problem takes as input a graph $G=(V,E)$ and a set $S \subseteq V$ and asks for the subset feedback vertex set of minimum cardinality.
In the weighted version every vertex of $G$ has a weight and the objective is to compute a subset feedback vertex set with the minimum total weight.
The \textsc{Subset Feedback Vertex Set} problem is a generalization of the classical \textsc{Feedback Vertex Set} problem in which the goal is to remove a set of vertices $X$ such that $G[V \setminus X]$ has no cycles.
Thus by setting $S = V$ the problem coincides with the NP-complete \textsc{Feedback Vertex Set} problem \cite{GJ}.
Both problems find important applications in several aspects that arise in optimization theory, constraint satisfaction, and bayesian inference \cite{EvenNZ00,fvs:approx:becker:1996,fvs:app:bayes:bar-yehuda:1998,problems:survey:festa:2009}.
Interestingly the \textsc{Subset Feedback Vertex Set} problem for $|S| = 1$ also coincides with the NP-complete \textsc{Multiway Cut} problem \cite{FominHKPV14} in which the task is to disconnect a predescribed set of vertices \cite{GargVY04,Calinescu08}.

{\sc Subset Feedback Vertex Set} was first introduced by Even et al.~who obtained a constant factor approximation algorithm for its weighted version~\cite{EvenNZ00}.
The unweighted version in which all vertex weights are equal has been proved to be fixed parameter tractable \cite{CyganPPW13}.
Moreover the fastest algorithm for the weighted version in general graphs runs in $O^{*}(1.86^n)$ time\footnote{The $O^{*}$ notation is used to suppress polynomial factors.} by enumerating its minimal solutions \cite{FominHKPV14},  whereas for the unweighted version the fastest algorithm runs in $O^{*}(1.75^n)$ time \cite{FominGLS16}.
As the unweighted version of the problem is shown to be NP-complete even when restricted to split graphs \cite{FominHKPV14}, there is a considerable effort to reduce the running time on chordal graphs, a proper superclass of split graphs, and more general on other classes of graphs.
Golovach et al. considered the weighted version and gave an algorithm that runs in $O^{*}(1.67^n)$ time for chordal graphs \cite{GolovachHKS14}.
Reducing the existing running time even on chordal graphs has been proved to be quite challenging and only for the unweighted version of the problem a faster algorithm was given that runs in $O^{*}(1.61^n)$ time \cite{ChitnisFLMRS13}.
In fact the $O^{*}(1.61^n)$-algorithm given in \cite{ChitnisFLMRS13} runs for every graph class which is closed under vertex deletions and edge contractions, and on which the weighted \textsc{Feedback Vertex Set} problem can be solved in polynomial time.
Thus there is an algorithm that runs in $O^{*}(1.61^n)$ time for the unweighted version of the {\sc Subset Feedback Vertex Set} problem when restricted to AT-free graphs \cite{ChitnisFLMRS13}, a graph class that properly contains permutation graphs and interval graphs.
Here we show that for the classes of permutation graphs and interval graphs we design a much faster algorithm even for the weighted version of the problem.

As a generalization of the classical \textsc{Feedback Vertex Set} problem, let us briefly give an overview of the complexity of \textsc{Feedback Vertex Set} on graph classes related to permutation graphs and interval graphs.
Concerning the complexity of \textsc{Feedback Vertex Set} on restricted graphs classes it is known to be NP-complete on bipartite graphs \cite{Yannakakis81a} and planar graphs \cite{GJ},
whereas it becomes polynomial-time solvable on the classes of
chordal graphs \cite{fvs:chord:corneil:1988,Spinrad03},
circular-arc graphs \cite{Spinrad03},
interval graphs \cite{fvs:int:lu:1997},
permutation graphs \cite{fvs:perm:brandstadt:1985,fvs:perm:brandstadt:1987,fvs:perm:brandstadt:1993,fvs:perm:liang:1994},
cocomparability graphs \cite{fvs:cocomp:liang:1997}, and, more generally, AT-free graphs \cite{KratschMT08}.
Despite the positive and negative results of the \textsc{Feedback Vertex Set} problem, very few similar results are known concerning the complexity of the \textsc{Subset Feedback Vertex Set} problem.
Clearly for graph classes for which the \textsc{Feedback Vertex Set} problem is NP-complete, so does the \textsc{Subset Feedback Vertex Set} problem.
However as the \textsc{Subset Feedback Vertex Set} problem is more general that \textsc{Feedback Vertex Set} problem, it is natural to study its complexity for graph classes for which \textsc{Feedback Vertex Set} is polynomial-time solvable.
In fact restricted to graph classes there is only a negative result for the \textsc{Subset Feedback Vertex Set} problem regarding the NP-completeness of split graphs \cite{FominHKPV14}.
Such a result, however, implies that there is an interesting algorithmic difference between the two problems, as the \textsc{Feedback Vertex Set} problem is known to be polynomial-time computable for chordal graphs \cite{fvs:chord:corneil:1988,Spinrad03}, and, thus, also for split graphs.

\begin{figure}[t]
\centering
\includegraphics[scale= 1.0]{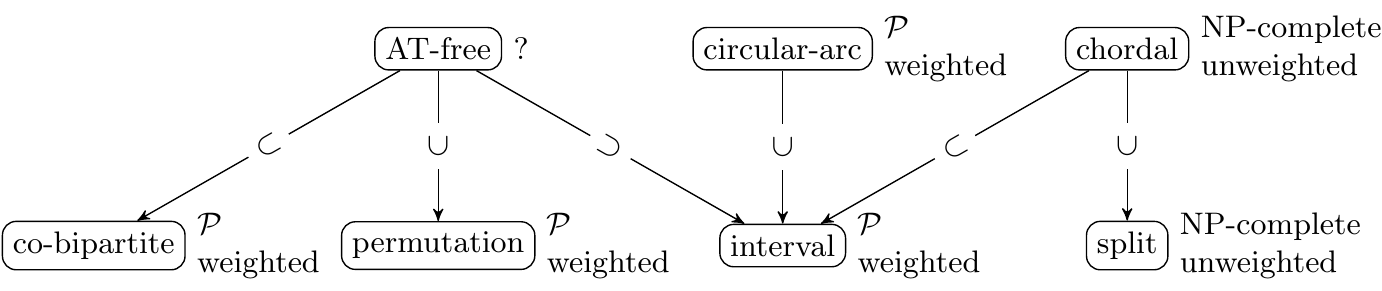}
\caption{The computational complexity of the {\sc Subset Feedback Vertex Set} problem restricted to the considered graph classes.
All polynomial-time results ($\mathcal{P}$) are obtained in this work, whereas the NP-completeness result of split graphs, and, consequently, of chordal graphs, is due to \cite{FominHKPV14}.}
\label{fig:classes}
\end{figure}

Here we initiate the study of \textsc{Subset Feedback Vertex Set} restricted on graph classes from the positive perspective.
We consider its weighted version and give the first positive results on permutation graphs and interval graphs, both being proper subclasses of AT-free graphs.
As already explained, we are interested towards subclasses of AT-free graphs since for chordal graphs the problem is already NP-complete \cite{FominHKPV14}.
Permutation graphs and interval graphs are unrelated to split graphs and are both characterized by a linear structure with respect to a given vertex ordering  \cite{graph:classes:brandstadt:1999,graph:classes:Go04,Spinrad03}.
For both classes of graphs we design polynomial-time algorithms based on dynamic programming of subproblems defined by passing the vertices of the graph according to their natural linear ordering.
One of our key ingredients is that during the pass of the dynamic programming we augment the considered vertex set and we allow the solutions to be chosen only from a specific subset of the vertices rather than the whole vertex set.
Although for interval graphs such a strategy leads to a simple algorithm, the case of permutation graphs requires further descriptions of the considered subsolutions by augmenting the considered part of the graph with a small number of additional vertices.
As a side result we show that the problem has a polynomial time solution on a larger class of interval graphs, namely to that of circular-arc graphs, by suitably applying our algorithm for interval graphs.
Moreover towards the unknown complexity of the problem for the class of AT-free graphs, we consider the class of co-bipartite graphs (complements of bipartite graphs) and settle the corresponding complexity status.
More precisely we show that the number of minimal solutions of a co-bipartite graph is polynomial which implies a polynomial-time algorithm of the \textsc{Subset Feedback Vertex Set} problem for the class of co-bipartite graphs.
Our overall results are summarized in Figure~\ref{fig:classes}.
Therefore, we contribute to provide the first positive results of the \textsc{Subset Feedback Vertex Set} problem on subclasses of AT-free graphs.

\section{Preliminaries}

All graphs in this text are undirected and simple. A graph is
denoted by $G=(V,E)$ with vertex set $V$ and edge set $E$. We use
the convention that $n=|V|$ and $m=|E|$. For a vertex subset $S
\subseteq V$, the {\it subgraph of $G$ induced by $S$} is $G[S] =
(S, \{\{u,v\} \in E \mid u,v \in S\})$. The {\it neighborhood} of
a vertex~$x$ of $G$ is $N(x)=\{v \mid xv \in E\}$ and the {\it
degree\/} of $x$ is $|N(x)|$. If $S \subseteq V$, then
$N(S)=\bigcup_{x \in S} N(x) \setminus S$. A \emph{clique} is a
set of pairwise adjacent vertices, while an \emph{independent set}
is a set of pairwise non-adjacent vertices.

A {\it path} is a sequence of distinct vertices $P=\langle v_1v_2\cdots v_k\rangle$ where each pair of consecutive vertices $v_iv_{i+1}$ forms an
edge of $G$. If additional $v_1v_k$ is an edge then we obtain a
{\it cycle}. In this paper, we distinguish between paths (or
cycles) and {\it induced paths} (or {\it induced cycles}). By an
induced path (or cycle) of $G$ we mean a chordless path (or
cycle). A graph is {\it connected} if there is a path between any
pair of vertices. A {\it connected component} of $G$ is a maximal
connected subgraph of $G$.
A {\it forest} is a graph that contains no cycles and a {\it tree}
is a forest that is connected.

\medskip

A {\it weighted graph} $G = (V, E)$ is a graph, where each vertex
$v \in V$ is assigned a {\it weight} that is a real number.
We denote by $w(v)$ the weight of each vertex $v \in V$.
For a vertex set $A \subset V$ the weight of $A$ is the sum of the
weights of all vertices in $A$.

The \textit{Subset Feedback Vertex Set} (SFVS) problem is defined
as follows: given a weighted graph $G$ and a vertex set $S\subseteq V$, find a vertex set $X \subset V$, such that all
cycles containing vertices of $S$, also contains a vertex of $X$
and $\sum_{v \in X} w(v)$ is minimized. In the unweighted version
of the problem all weights are equal. A vertex set $X$ is defined
as \textit{minimal} subset feedback vertex set if no proper subset
of $X$ is a subset feedback vertex set for $G$ and $S$. The
classical \textsc{Feedback Vertex Set} (FVS) problem is a special
case of the subset feedback vertex set problem with $S = V$.
Note that a {\it minimum weight} subset feedback vertex set is dependent on the
weights of the vertices, whereas a {\it minimal} subset feedback vertex
set is only dependent on the vertices and not their weights. Clearly, both
in the weighted and the unweighted versions, a minimum subset feedback
vertex set must be minimal.

An induced cycle of $G$ is called $S$-cycle if a vertex of $S$ is contained in the cycle.
We define an \emph{$S$-forest} of $G$ to be a vertex set $Y \subseteq V$ such that no cycle in $G[Y]$ is an $S$-cycle.
An $S$-forest $Y$ is \textit{maximal} if no proper superset of $Y$ is an $S$-forest.
Observe that $X$ is a minimal subset feedback vertex set if and only if $Y = V \setminus X$ is a maximal $S$-forest.
Thus, the problem of computing a minimum weighted subset feedback vertex set is equivalent to the problem of computing a maximum weighted $S$-forest.
Let us denote by $\mathcal{F}_S$ the class of $S$-forests.
In such terms, given the graph $G$ and the subset $S$ of $V$, we are interested in finding a $\max_w\left\{Y \subseteq V \mid  G[Y] \in \mathcal{F}_S \right\}$, where $\max_w$ selects a vertex set having the maximum sum of its weights.


\section{Computing SFVS on interval graphs and circular-arc graphs}
Here we present a polynomial-time algorithm for the SFVS problem
on interval graphs. 
A graph is an \emph{interval graph} if there is a bijection between its vertices and a
family of closed intervals of the real line such that two vertices are adjacent if and
only if the two corresponding intervals overlap.
Such a bijection is called an \emph{interval representation} of the graph, denoted by $\mathcal{I}$.
Notice that every induced subgraph of an interval graph is an interval graph.
Hereafter we assume that the input graph is connected; otherwise, we apply the described algorithm in each connected component and take the overall solution as the union of the sub-solutions.

As already mentioned, instead of finding a subset feedback vertex
set $X$ of minimum weight of $(G,S)$ we concentrate on the equivalent problem of finding a
maximum weighted $S$-forest $Y$ of $(G,S)$.
We first define the necessary vertex sets.
Let $G$ be a weighted interval graph and let $\mathcal{I}$ be its interval representation.
The left and right endpoints of an interval $i$, $1 \leq i \leq n$, are denoted by $\ell(i)$ and $r(i)$, respectively.
Each interval $i$ is labeled from 1 to $n$ according to their ascending $r(i)$.
For technical reasons of our algorithm, we add an interval with label $0$ that does not belong to $S$, has negative weight, and augment $\mathcal{I}$ to $\mathcal{I}^{+}$ by setting $\ell(0)=-1$ and $r(0)=0$.
Notice that interval 0 is non-adjacent to any vertex of $G$.
Clearly if $Y$ is a maximum weighted $S$-forest for $G[\mathcal{I}^{+}]$ then $Y \setminus \{0\}$ is a maximum weighted $S$-forest for $G[\mathcal{I}]$.
Moreover it is known that any induced cycle of an interval graph is an induced triangle \cite{fvs:int:lu:1997,Spinrad03}.

We consider the two relations on $V$ that are defined by the endpoints of the intervals as follows:
\begin{displaymath}
\begin{array}{lll}
i\leq_{\ell}j& \Longleftrightarrow & \ell(i)\leq \ell(j)\\
i\leq_{r}j& \Longleftrightarrow & r(i)\leq r(j)
\end{array}
\end{displaymath}
Since all endpoints of the collection's intervals are distinct, it is not difficult to show that $\leq_{\ell}$ and $\leq_{r}$ are total orders on $V$.
We write $\lmin$ to denote the interval having the minimum left-endpoint among its operands with respect to $\leq_{\ell}$;
and we write $\rmax$ to denote the interval having the maximum right-endpoint among its operands with respect to $\leq_{r}$.
We define two different types of predecessors of an interval $i$ with respect to $\leq_{r}$, which correspond to the subproblems that our dynamic programming algorithm wants to solve.
Let $i\in V\setminus\{0\}$. Then
\begin{itemize}
\item ${<}i=\df \rmax\{h\in V:h<_{r}i\}$ and
\item ${\ll} i=\df \rmax\{h\in V:h<_{r}i\textrm{ and }\{h,i\}\notin E\}$.
\end{itemize}
Moreover for a vertex $i\in V$ we let $V_{i}=\df \{h\in V:h\leq_{r}i\}$.
Observe that for two vertices $i,x \in V$ with $r(i)<r(x)$, $x \in V \setminus V_i$.
An example of interval graph with the given notation of $V_i$ is given in Figure~\ref{fig:inerval}.

\begin{figure}[t]
\centering
\includegraphics[scale= 0.93]{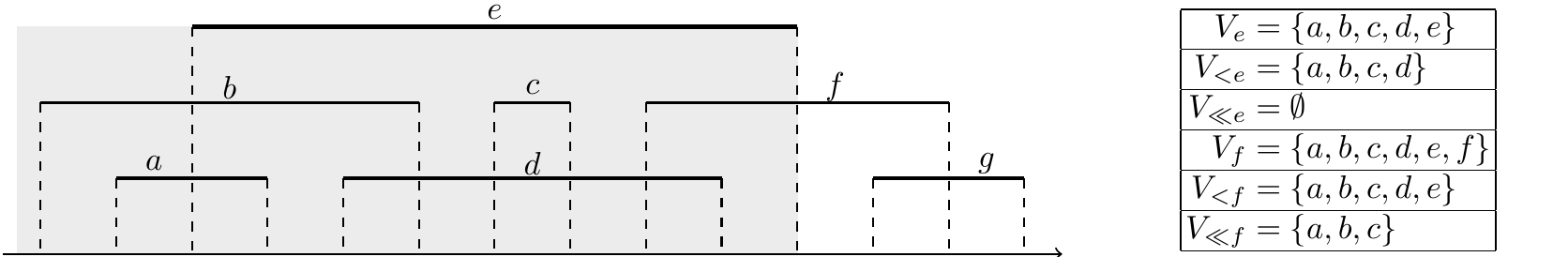}
\caption{An interval graph given by its interval representation and the corresponding sets of $V_{e}$ and $V_{f}$.
Observe that ${<}f = e$ whereas ${\ll}f=c$. Also notice that the intervals that are properly contained within the gray area form the set $V_{e}$.}
\label{fig:inerval}
\end{figure}

\begin{observation}\label{obs:sfvs:int}
Let $i\in V\setminus\{0\}$ and let $j\in V\setminus V_{i}$ such that $\{i,j\}\in E$. Then,
\begin{enumerate}[label={(\arabic*)}]
\item $V_{i}=V_{<i}\cup\{i\}$ and
\item $V_{<i}=V_{\ll j}\cup\{h\in V_{<i}:\{h,j\}\in E\}$.
\end{enumerate}
\end{observation}
\begin{proof}
The first statement follows by the definitions of $V_{i}$ and ${<}i$.
For the second statement observe that $V_{<i}$ can be partitioned into the non-neighbors of $j$ in $V_{<i}$ and the neighbors of $j$ in $V_{<i}$.
The first set corresponds to $V_{\ll j}$ whereas the second set is exactly the set $\{h\in V_{<i}:\{h,j\}\in E\}$.
\end{proof}

Now we define the sets that our dynamic programming algorithm uses in order to compute the $S$-forest-inducing vertex set of $G$ that has maximum weight.
\begin{definition}[$A$-sets]\label{def:sfvs:int:a}
Let $i\in V$. Then,
\begin{displaymath}
A_{i}=\df\max_{w}\{X\subseteq V_{i}:G[X]\in\mathcal{F}_{S}\}.
\end{displaymath}
\end{definition}
\begin{definition}[$B$-sets]\label{def:sfvs:int:b}
Let $i\in V$ and let $x\in V\setminus V_{i}$. Then,
\begin{displaymath}
B_{i}^{x}=\df\max_{w}\{X\subseteq V_{i}:G[X\cup\{x\}]\in\mathcal{F}_{S}\}.
\end{displaymath}
\end{definition}
\begin{definition}[$C$-sets]\label{def:sfvs:int:c}
Let $i\in V$ and let $x,y\in V\setminus(V_{i}\cup S)$ such that $x<_{\ell}y$ and $\{x,y\}\in E$. Then,
\begin{displaymath}
C_{i}^{x,y}=\df\max_{w}\{X\subseteq V_{i}:G[X\cup\{x,y\}]\in\mathcal{F}_{S}\}.
\end{displaymath}
\end{definition}

Since $V_{0}=\{0\}$ and $w(0)\leq0$, $A_{0}=\varnothing$ and, since $V_{n}=V$, $A_{n}=\max_{w}\{X\subseteq V:G[X]\in\mathcal{F}_{S}\}$.
The following lemmas state how to recursively compute all $A$-sets, $B$-sets and $C$-sets besides $A_{0}$.

\begin{lemma}\label{lem:sfvs:int:i}
Let $i\in V\setminus\{0\}$. Then $A_{i}=\max_{w}\left\{A_{<i},B_{<i}^{i}\cup\{i\}\right\}$.
\end{lemma}
\begin{proof}
By~Observation \ref{obs:sfvs:int}~(1), $V_{i}=V_{<i}\cup\{i\}$.
There are two cases to consider for $i\in V\setminus\{0\}$: either $i\notin A_{i}$ or $i\in A_{i}$.
In the former we have $A_{i}=A_{<i}$, whereas in the latter $i$ cannot induce an $S$-cycle in $B_{<i}^{i}$ by definition, which implies that $A_{i}=B_{<i}^{i}\cup\{i\}$.
\end{proof}

For a set of vertices $L \subseteq V$, the {\it leftmost} vertex is the vertex of $L$ having the minimum left endpoint.
That is, the leftmost vertex of $L$ is the vertex $x'$ such that $x' = \lmin\{L\}$.

\begin{lemma}\label{lem:sfvs:int:ix}
Let $i\in V$ and let $x\in V\setminus V_{i}$. 
Moreover, let $x'$ be the leftmost vertex of $\{i,x\}$ and let $y'$ be the vertex of $\{i,x\}\setminus \{x'\}$.
\begin{enumerate}[label={(\arabic*)}]
\item If $\{i,x\}\notin E$, then $B_{i}^{x}=A_{i}$.
\item If $\{i,x\}\in E$, then
$
B_{i}^{x}=\left\{\begin{array}{l@{}l}
\max_{w}\left\{B_{<i}^{x},B_{\ll y'}^{x'}\cup\{i\}\right\}&,\textrm{ if }i\in S\textrm{ or }x\in S\\
\max_{w}\left\{B_{<i}^{x},C_{<i}^{x',y'}\cup\{i\}\right\}&,\textrm{ if }i,x\notin S.\\
\end{array}\right.
$
\end{enumerate}
\end{lemma}
\begin{proof}
Assume first that $\{i,x\}\notin E$. Then $r(i)<\ell(x)$, so that $x$ has no neighbour in $G[V_{i}\cup\{x\}]$.
Thus no subset of $V_{i}\cup\{x\}$ containing $x$ induces an $S$-cycle of $G$.
By Definitions \ref{def:sfvs:int:a} and \ref{def:sfvs:int:b}, it follows that $B_{i}^{x}=A_{i}$.

Next assume that $\{i,x\}\in E$.
If $i\notin B_{i}^{x}$ then according to Observation~\ref{obs:sfvs:int}~(1) it follows that $B_{i}^{x}=B_{<i}^{x}$.
So let us assume in what follows that $i\in B_{i}^{x}$.
Observe that $B_{i}^{x}\setminus\{i\}\subseteq V_{<i}$, by Observation~\ref{obs:sfvs:int}~(1).
We distinguish two cases according to whether $i$ or $x$ belong to $S$.
\begin{itemize}
\item Let $i\in S$ or $x\in S$.
Assume there is a vertex $h\in B_{i}^{x}\setminus\{i\}$ such that $\{h,y'\}\in E$.
Then we know that $\ell(y')< r(h)$ and by definition we have $\ell(x')<\ell(y')$ and $r(h) < r(x')$.
This particularly means that $h$ is adjacent to $x'$.
This however leads to a contradiction since $\langle h,x',y'\rangle$ is an induced $S$-triangle of $G$.
Thus for any vertex $h\in B_{i}^{x}\setminus\{i\}$ we know that $\{h,y'\}\notin E$.
By Observation~\ref{obs:sfvs:int}~(2) notice that $B_{i}^{x}\setminus\{i\}\subseteq V_{\ll y'}$.
Also observe that the neighbourhood of $y'$ in $G[V_{\ll y'}\cup\{x',y'\}]$ is $\{x'\}$.
Thus no subset of $V_{\ll y'}\cup\{x',y'\}$ that contains $y'$ induces an $S$-cycle of $G$.
Therefore $B_{i}^{x}=B_{\ll y'}^{x'}\cup\{i\}$.

\item Let $i,x\notin S$. By the fact $V_{i}=V_{<i}\cup\{i\}$ and since $x' <_{\ell} y'$ we get
$B_{i}^{x}=C_{<i}^{x',y'}\cup\{i\}$.
\end{itemize}
Therefore in all cases we reach the desired equations.
\end{proof}

\begin{lemma}\label{lem:sfvs:int:ixy}
Let $i\in V$ and let $x,y\in V\setminus(V_{i}\cup S)$ such that $x<_{\ell}y$ and $\{x,y\}\in E$.
Moreover, let $x'$ be the leftmost vertex of $\{i,x,y\}$ and let $y'$ be the vertex of $\{i,x,y\}\setminus \{x'\}$.
\begin{enumerate}
\item If $\{i,y\}\notin E$, then $C_{i}^{x,y}=B_{i}^{x}$.
\item If $\{i,y\}\in E$, then $C_{i}^{x,y}=\left\{\begin{array}{l@{}l}
C_{<i}^{x,y}&,\textrm{ if }i\in S\\
\max_{w}\left\{C_{<i}^{x,y},C_{<i}^{x',y'}\cup\{i\}\right\}&,\textrm{ if }i\notin S.\\
\end{array}\right.$
\end{enumerate}
\end{lemma}
\begin{proof}
Assume first that $\{i,y\}\notin E$. Then $r(i),\ell(x)<\ell(y)<r(x)$, so that the neighbourhood of $y$ in $G[V_{i}\cup\{x,y\}]$ is $\{x\}$.
Thus no subset of $V_{i}\cup\{x,y\}$ that contains $y$ induces an $S$-cycle of $G$.
By Definitions~\ref{def:sfvs:int:b} and \ref{def:sfvs:int:c}, it follows that $C_{i}^{x,y}=B_{i}^{x}$.

Assume next that $\{i,y\}\in E$. Then $\ell(x)<\ell(y)<r(i)<r(x),r(y)$, so that $\langle i,x,y\rangle$ is an induced triangle of $G$.
If $i\notin C_{i}^{x,y}$ then by Observation \ref{obs:sfvs:int}~(1) we have $C_{i}^{x,y}=C_{<i}^{x,y}$.
Suppose that $i\in C_{i}^{x,y}$.
If $i\in S$ then $\langle i,x,y\rangle$ is an induced $S$-triangle of $G$, contradicting the fact that $i\in C_{i}^{x,y}$.
By definition observe that $x,y \in C_{i}^{x,y}$ which means that also $x \notin S$ and $y \notin S$.
Hence $i\in C_{i}^{x,y}$ implies that $S \cap \{i,x,y\} = \emptyset$.
We will show that under the assumptions $\{i,y\}\in E$ and $i\in C_{i}^{x,y}$, we have $C_{i}^{x,y}=C_{<i}^{x',y'}\cup\{i\}$.
Notice that $C_{i}^{x,y}\setminus\{i\}\subseteq V_{<i}$ which means that every solution of $C_{i}^{x,y}$ is a solution of $C_{<i}^{x',y'}\cup\{i\}$.

To complete the proof we show that every solution of $C_{<i}^{x',y'}\cup\{i\}$ is indeed a solution of $C_{i}^{x,y}$.
Let $z'$ be the vertex of $\{i,x,y\}\setminus \{x',y'\}$.
Observe that by the leftmost ordering we have $\ell(x')<\ell(y')<\ell(z')$.
We consider the graph $G_{<i}$ induced by the vertices of $V_{<i}\cup\{x',y',z'\}$.
Assume for contradiction that an $S$-triangle of $G_{<i}$ is not an $S$-triangle in $C_{i}^{x',y'}$.
Every $S$-triangle that contains $x'$ or $y'$ of $G_{<i}$ remains an $S$-triangle in $C_{i}^{x',y'}$.
Thus $z'$ must be contained in such an $S$-triangle of $G_{<i}$.
Let $\langle v_{1},v_{2},z'\rangle$ be an induced $S$-triangle of $G_{<i}$, where $v_1,v_2 \in V_{<i}$.
Since $x',y',z'\notin S$, without loss of generality, assume that $v_{1}\in S$.
The $S$-triangle of $G_{<i}$ implies that $\ell(z')<r(v_{1})$.
By the fact that $v_1 \in \in V_{<i}$ we have $r(v_{1})<r(x'),r(y'),r(z')$.
Since $\ell(x')<\ell(y')<\ell(z')$ the previous inequalities imply that $\{v_{1},x'\},\{v_{1},y'\}\in E$.
Thus $\langle v_{1},x',y'\rangle$ is an induced $S$-triangle in $C_{i}^{x',y'}$, leading to a contradiction.
Therefore $C_{i}^{x,y}=C_{<i}^{x',y'}\cup\{i\}$ as desired.
\end{proof}

Now we are equipped with our necessary tools to obtain the main result of this section, namely a polynomial-time algorithm for SFVS on interval graphs.
\begin{theorem}\label{theo:interval}
{\sc Subset Feedback Vertex Set} can be solved in $O(n^3)$ time on interval graphs.
\end{theorem}
\begin{proof}
We briefly describe such an algorithm based on Lemmas~\ref{lem:sfvs:int:i}, \ref{lem:sfvs:int:ix}, and \ref{lem:sfvs:int:ixy}.
In a preprocessing step we compute $<i$ and $\ll i$ for each interval $i\in V\setminus\{0\}$.
We scan all intervals from $0$ to $n$ in an ascending order with respect to $<_{\ell}$.
For every interval $i$ that we visit, we compute first $A_{i}$ according to Lemma~\ref{lem:sfvs:int:i} and
then compute $B_{i}^{x}$ and $C_{i}^{x,y}$ for every $x,y$ such that $\ell(i) <\ell(x) < \ell(y)$ according to Lemmas~\ref{lem:sfvs:int:ix} and \ref{lem:sfvs:int:ixy}, respectively.
At the end we output $A_{n}$ as already explained.
The correctness of the algorithm follows from Lemmas~\ref{lem:sfvs:int:i}--\ref{lem:sfvs:int:ixy}.

Regarding the running time, notice that computing $<i$ and $\ll i$ can be done in $O(n)$ time since the intervals are sorted with respect to their end-points.
The computation of a single $A$-set, $B$-set or $C$-set takes constant time.
Therefore the overall running time of the algorithm is $O(n^3)$.
\end{proof}

We will show that the previous algorithm can be
polynomially applied for a superclass of interval graphs, namely
to that of circular-arc graphs.
A graph is called {\it circular-arc} if it is the intersection graph of arcs of a circle.
Every vertex is represented by an arc, such that two vertices are adjacent if and only if the corresponding arcs intersect.
The family of arcs corresponding to the vertices of the graph constitute a {\it circular-arc model}.

\begin{theorem}\label{theo:circular}
{\sc Subset Feedback Vertex Set} can be solved in $O(n^4)$ time on circular-arc graphs.
\end{theorem}
\begin{proof}
Let $G$ be the input circular-arc graph and let $S \subseteq V$.
If $G$ is an interval graph then we run the algorithm given in Theorem~\ref{theo:interval}.
Otherwise in any circular-arc model the whole circle is covered by arcs.
Let the set of $n$ arcs of the circle have arc end-points $1, \ldots, 2n$.
We denote by $A_i$ the set of arcs containing point $i$.
By the circular-arc model notice that for each $i$ the graph $G-A_i$ is an interval graph.
Moreover notice that the subgraph induced by $A_i$ is a clique and, thus, an interval graph.

Our algorithm proceeds in two phases.
First it chooses all $2n$ points and for each removes all arcs containing it.
Then we run the interval graph algorithm on the remaining graph without the arcs of the point. 
More formally, for each point $i$ let $X_i$ be a maximum $S$-forest of $G-A_i$.
Observe that $X_i$ can be computed by the interval algorithm.
At the end we choose the maximum weighted set among all $X_i$ for each point $i$.
This constitutes the first phase of our algorithm.

Concerning its correctness, recall that the whole circle is covered by arcs.
Let $\mathcal{C}$ be the family of minimal set of arcs that cover the whole circle.
For any set $C \in \mathcal{C}$ observe that $G[C]$ is a cycle.
Assume first that there is an $S$-cycle $C_S \in \mathcal{C}$.
Then we must remove for some point of the circle all arcs containing it.
Thus if there is an $S$-cycle $C_S \in \mathcal{C}$ then there is a point $i$ such that the set $X_i$ is a maximum $S$-forest.

In what follows we assume that there is no $S$-cycle in $\mathcal{C}$.
Consider an $S$-cycle $C'_S$ in $G$.
We show next that $C'_S \subseteq A_i$, for some point $i$.
Suppose for contradiction that the $S$-cycle $C'_S \nsubseteq A_i$.
Choose an appropriate $A_i$ such that $s_i \in A_i \cap S$.
Then there is a neighbour $x$ of $s_i$ that belongs to $C'_S$ such that $x \notin A_i$.
Let $y$ be the other neighbour of $s_i$ in the $S$-cycle $C'_S$.
If $x$ and $y$ are adjacent then $\{x,y,s_i\}$ induce an $S$-triangle and there is a point $i'$ such that $\{x,y,s_i\} \subseteq A_{i'}$.
Otherwise, there is a path between $x$ and $y$ because of the $S$-cycle $C'_S$.
Such a path is contained in the graph $G-A_i$ so that there is a cycle $C_S \in \mathcal{C}$, leading to a contradiction.
Therefore such an $S$-cycle is completely contained in some $A_i$.

At the second phase we remove any vertex that is non-adjacent to an $S$-vertex.
Let $N_S$ be the vertices that are non-adjacent to any $S$-vertex.
In any solution $X$ of the graph $G - N_S$ we can safely add the vertices of $N_S$.
This is because no vertex of $N_S$ participates in any $S$-cycle $C'_S$ of $G$, since $C'_S \subseteq A_i$ for some point $i$.

\begin{figure}[t]
\centering
\includegraphics[scale= 1.0]{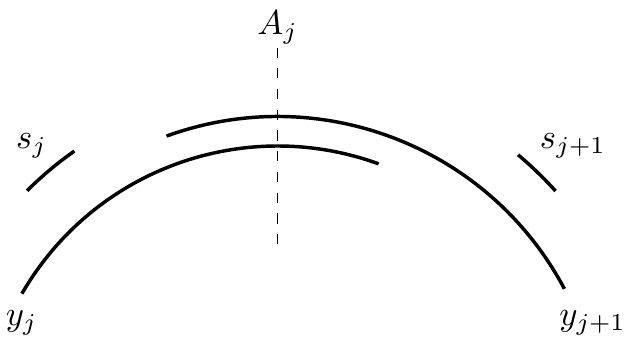}
\caption{A circular-arc graph given by its circular-arc model considered in the proof of Theorem~\ref{theo:circular}.}
\label{fig:circular}
\end{figure}

If the graph $G-N_S$ is an interval graph then we run the interval algorithm given in Theorem~\ref{theo:interval}.
Otherwise, we will show that there is an edge $e$ in $G-N_S$ such that
\vspace*{-0.1in}
\begin{itemize}
\item[(i)] no $S$-cycle of $G-N_S$ passes through $e$ and
\vspace*{-0.1in}
\item[(ii)] the graph obtained from $G-N_S$ by removing $e$ is an interval graph.
\end{itemize}
\vspace*{-0.1in}
Since $G-N_S$ is a circular-arc graph but not an interval graph, $\mathcal{C}$ is non-empty.
Let $\langle y_1, y_2, \ldots, y_k \rangle$ be a chordless cycle of $\mathcal{C}$ with $k \geq 4$.
For every two adjacent vertices $y_{j},y_{j+1}$ of the cycle there is a point $j$ such that $y_{j},y_{j+1} \in A_{j}$, $y_j \notin A_{j+1}$, and $y_{j+1} \notin A_{j-1}$.
Figure~\ref{fig:circular} shows the corresponding situation.
Then observe that $A_j \cap S = \emptyset$, since for any vertex $x\in A_j \cap S$ we have an $S$-cycle $\langle y_1, \ldots, y_j,x,y_{j+1},\ldots, y_k \rangle$ in $\mathcal{C}$.
Let $s_j$ and $s_{j+1}$ be the first counterclockwise and clockwise arcs, respectively, that belong to $S$ starting at point $j$.
It is clear that $y_j$ is adjacent to $s_j$ and $y_{j+1}$ is adjacent to $s_{j+1}$.
If $y_j$ is also adjacent to $s_{j+1}$ then $\langle y_1, \ldots, y_j,s_{j+1},y_{j+1},\ldots, y_k \rangle$ is an $S$-cycle in $\mathcal{C}$.
Thus $y_j$ is adjacent to $s_j$ and non-adjacent to $s_{j+1}$ whereas $y_{j+1}$ is adjacent to $s_{j+1}$ and non-adjacent to $s_j$.

Next we show that $A_j = \{y_j,y_{j+1}\}$.
Let $x \in A_j \setminus \{y_j,y_{j+1}\}$.
Since every vertex of $G-N_S$ is adjacent to at least one vertex of $S$, $x$ is adjacent to $s_j$ or $s_{j+1}$.
If $x$ is adjacent to $s_j$ then there is an $S$-cycle $\langle y_1, \ldots, y_j,s_{j}, x, y_{j+1}, \ldots, y_k \rangle$ in $\mathcal{C}$. 
The symmetric $S$-cycle occurs whenever $x$ is adjacent to $s_{j+1}$. 
Thus in all cases we reach a contradiction so that $A_j = \{y_j,y_{j+1}\}$. 
This shows that if we remove the edge $e'=\{y_j,y_{j+1}\}$ from $G-N_S$ then we get an interval graph by the circular-arc model of $G-N_S$.
Moreover no $S$-cycle passes through the edge $e'=\{y_j,y_{j+1}\}$ as then we would have an $S$-cycle in $\mathcal{C}$.
Hence by Theorem~\ref{theo:interval} we run the algorithm on the graph obtained from $G-N_S$ by removing the edge $e'$ and compute a solution $X'$.
Therefore $X' \cup N_S \cup e'$ is a solution of the original graph $G$.

Regarding the running time the most time consuming is the first phase of our algorithm in which we need to run for every point the interval algorithm.
Observe that there are at most $2n$ points and for each we need to run the $O(n^3)$ interval algorithm.
Therefore the total running time of our algorithm is $O(n^4)$.
\end{proof}

\section{Computing SFVS on permutation graphs}
Let $\pi = \pi(1), \ldots, \pi(n)$ be a permutation over $\{1, \ldots, n\}$, that is a bijection between $\{1, \ldots, n\}$ and $\{1, \ldots, n\}$.
The position of an integer $i$ in $\pi$ is denoted by $\pi^{-1}(i)$.
Given a permutation $\pi$, the {\it inversion graph} of $\pi$, denoted by $G(\pi)$, has vertex set $\{1, \ldots, n\}$ and two vertices $i,j$ are adjacent if $(i-j)(\pi(i)-\pi(j))<0$.
A graph is a {\it permutation graph} if it is isomorphic to the inversion graph of a permutation \cite{graph:classes:brandstadt:1999,graph:classes:Go04}.
For our purposes, we assume that a permutation graph is given as a permutation $\pi$ and equal to the defined inversion graph.
Permutation graphs also have an interesting geometric intersection model: they are the intersection graphs of segments between two horizontal parallel lines, that is, there is a one-to-one mapping from the segments onto the vertices of a graph such that there is an edge between two vertices of the graph if and only if their corresponding segments intersect.
We refer to the two horizontal lines as {\it top} and {\it bottom} lines.
This representation is called a {\it permutation diagram} and a graph is a permutation graph if and only if it has a permutation diagram.
It is important to note that every induced subgraph of a permutation graph is a permutation graph.
Every permutation graph with permutation $\pi$ has a permutation diagram in which the endpoints of the line segments on the bottom line appear in the same order as they appear in $\pi$.

We assume that we are given a connected permutation graph $G=(V,E)$ such that $G=G(\pi)$ along with a vertex set $S \subseteq V$ and a weight function $w: V \rightarrow \mathbb{R}^{+}$ as input.
We add an isolated vertex in $G$ and augment $\pi$ to $\pi'$ as follows: $\pi' = \{0\} \cup \pi$ with $\pi'(0)=0$.
Further we assign a negative value for $0$'s weight and assume that $0 \notin S$.
It is important to note that any induced cycle of a permutation graph is either an induced triangle or an induced square \cite{fvs:perm:brandstadt:1985,fvs:perm:brandstadt:1993,fvs:perm:brandstadt:1987,fvs:perm:liang:1994,Spinrad03}.

We consider the two relations on $V$ defined as follows: $i\leq_{t}j$ if and only if $i\leq j$ and $i\leq_{b}j$ if and only if $\pi^{-1}(i)\leq\pi^{-1}(j)$ for all $i,j\in V \cup \{0\}$.
It is not difficult to see that both $\leq_{t}$ and $\leq_{b}$ are total orders on $V$; they are exactly the orders in which the integers appear on the top and bottom line, respectively, in the permutation diagram.
Moreover we write $i <_{t} j$ or $i <_{b} j$ if and only if $i \neq j$ and $i \leq_{t} j$ or $i \leq_{b} j$, respectively.
We extend $\leq_{t}$ and $\leq_{b}$ to support sets of vertices as follows.
For two sets of vertices $L$ and $R$ we write $L \leq_{t} R$ (resp., $L \leq_{b} R$) if for any two vertices $u \in L$ and $v \in R$, $u \leq_{t} v$ (resp., $u \leq_{b} v$).

Let $G=\df G(\pi)$ be a permutation graph.
Two vertices $i,j\in \{0, 1, \ldots, n\}$ with $i \leq_{t} j$ are called {\it crossing pair}, denoted by $ij$, if $j \leq_{b} i$.
We denote by $\mathcal{X}$ the set of crossing pairs in $G$.
Observe that a crossing pair $ij$ with $i \neq j$ corresponds to an edge of $G$.
In order to distinguish the edges of $G$ with the crossing pairs of the form $ii$ we let $\mathcal{I} = \left\{ii \mid i \in \{1, \ldots, n\}\right\}$, so that $\mathcal{X} \setminus \mathcal{I}$ contains exactly the edges of $G$.

Given two crossing pairs $gh, ij \in \mathcal{X}$ we define two partial orderings:
\begin{displaymath}
\begin{array}{lll}
gh \leq_{\ell} ij & \Longleftrightarrow & g \leq_{t} i \text{ and } h \leq_{b} j\\
gh \leq_{r} ij & \Longleftrightarrow & g \leq_{b} i \text{ and } h \leq_{t} j.
\end{array}
\end{displaymath}
As in the case for interval graphs, we write $\lmin$ to denote the crossing pair with the minimum top and bottom integers with respect to $\leq_{\ell}$;
and we write $\rmax$ to denote the crossing pair with the maximum top and bottom integers with respect to $\leq_{r}$.

Our dynamic programming algorithm iterates on ordered crossing vertex pairs.
We next define the predecessors of a crossing pair with respect to $\leq_{r}$, which correspond to the subproblems that our dynamic programming algorithm wants to solve.
Let $ij\in\mathcal{X}\setminus\{00\}$ be a crossing pair.
We define the set of vertices that induce the part of the subproblem that we consider at each crossing pair as follows: $V_{ij}=\df\{h\in V:hh\leq_{r}ij\}$.
Let $x$ be a vertex such that $i <_{b} x$ or $j <_{t} x$.
By definition notice that $x$ does not belong in $V_{ij}$.
The predecessors of the crossing pair $ij$ are defined as follows:
\begin{itemize}
\item ${\eqslantless}ij =\df \rmax\{gh\in\mathcal{X}:gh\leq_{r}ij\textrm{ and }h\neq j\}$,
\item ${\leqslant} ij=\df\rmax\{gh\in\mathcal{X}:gh\leq_{r}ij\textrm{ and }g\neq i\}$,
\item ${<}ij =\df\rmax\{gh\in\mathcal{X}:gh<_{r}ij$,
\item ${\ll}ij =\df\rmax\{gh\in\mathcal{X}:gh<_{r}ij\textrm{ and }\{g,i\},\{g,j\},\{h,i\},\{h,j\}\notin E\}$, and
\item ${<}ij{\ll}xx =\df\rmax\{gh\in \mathcal{X}:gh<_{r}ij\textrm{ and }\{g,x\},\{h,x\}\notin E\}$.
\end{itemize}
Although it seems somehow awkward to use one the symbols $\left\{{\eqslantless}, {\leqslant}, {<}, {\ll}, {<} {\ll}\right\}$ for the defined predecessors,
we stress that such predecessors are required only to describe the necessary subset $V_{gh}$ of $V_{ij}$.
Moreover it is not difficult to see that each of the symbol gravitates towards a particular meaning with respect to the top and bottom orderings as well as the non-adjacency relationship.
An example of a permutation graph that illustrates the defined predecessors is given in Figure~\ref{fig:permutation}.

\begin{figure}[t]
\centering
\includegraphics[scale= 1.0]{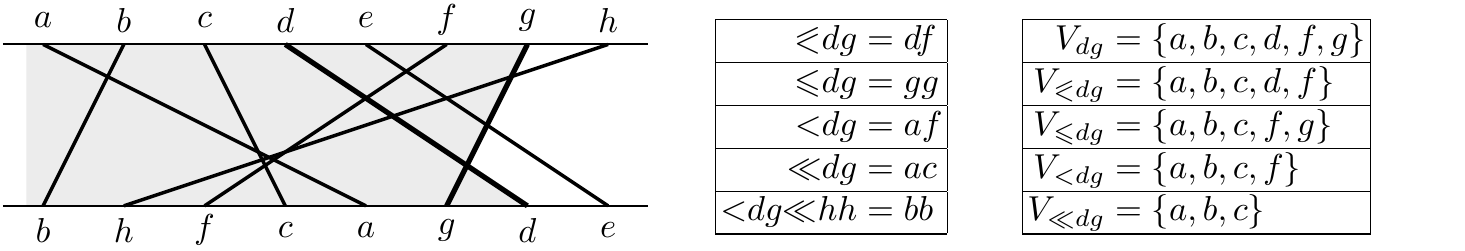}
\caption{A permutation graph given by its permutation diagram and the set $V_{dg}$ of the crossing pair $dg$ together with the corresponding predecessors of $dg$.
Observe that the line segments that are properly contained within the gray area form the set $V_{dg}$.}
\label{fig:permutation}
\end{figure}

With the above defined predecessors of $ij$, we show how the set $V_{ij}$ can be partitioned into smaller sets of vertices with respect to a suitable predecessor.

\begin{observation}\label{obs:sfvs:perm}
Let $ij\in\mathcal{X}$ and let $x\in V\setminus V_{ij}$. Then,
\begin{enumerate}[label={(\arabic*)}]
\item $V_{ij}=V_{{\eqslantless}ij}\cup\{j\} = V_{{\leqslant}ij}\cup\{i\} = V_{{<ij}}\cup\{i,j\}$,
\item $V_{{<}ij}=V_{{\ll}jj}\cup\{h\in V_{{<}ij}:\{h,j\}\in E\} = V_{{\ll} ii}\cup\{h\in V_{{<}ij}:\{h,i\}\in E\}$,
\item $V_{{\ll} ii}=V_{{\ll} ij}\cup\{h\in V_{{\ll} ii}:\{h,j\}\in E\}$,
\item $V_{{\ll} jj}=V_{{\ll} ij}\cup\{h\in V_{{\ll} jj}:\{h,i\}\in E\}$, and
\item $V_{{<}ij}=V_{{<}ij{\ll} xx}\cup\{h\in V_{{<}ij}:\{h,x\}\in E\}$.
\end{enumerate}
\end{observation}
\begin{proof}
Let $ij_1$ be the predecessor ${\eqslantless}ij$.
By the $\rmax$ choice of $j_1$, there is no vertex $j'$ such that $j_1 <_{t} j' <_{t} j$. 
Thus $V_{ij_1} \cup \{j\}$ is the set $V_{ij}$.
The rest of the equalities in the first statement follow in a similar way.

Let $i_1j_1$ be the predecessor ${\ll}jj$. Then both $i_1$ and $j_1$ are non-adjacent to $j$ and have the maximum values such that $i_1 <_{b} j$ and $j_1 <_{t} j$, respectively.
This particularly means that $i_1 <_{t} j_1 <_{t} j$ and $j_1 <_{b} i_1 <_{b} j$.
Thus any vertex $i' \in V_{ij} \setminus \{i,j\}$ with $j_1 <_{t} i' <_{t} j$ or $i_1 <_{b} i' <_{b} j$ must be adjacent to $j$
which implies that $V_{{<}ij} \setminus V_{{\ll}jj}$ contains exactly the neighbours of $j$ in $V_{{<}ij}$.
These arguments imply the second, third, and fourth statements.

For the last statement, notice that $V_{{<}ij}$ can be partitioned into the neighbours and the non-neighbours of $x$.
By definition $V_{{<}ij{\ll} xx}$ contains the non-neighbours of $x$ so that every vertex of $V_{{<}ij} \setminus V_{{<}ij{\ll} xx}$ is adjacent to $x$.
\end{proof}

It is clear that for any edge $\{i,j\} \in E$ either $i <_t j$ and $j <_b i$ hold, or $j <_t i$ and $i <_b j$ hold.
If further $ij \in \mathcal{X}\setminus\mathcal{I}$, is a crossing pair then we know that $i <_t j$ and $j <_b i$.

Next we define the sets that our dynamic programming algorithm computes in order to compute the induced $S$-forest vertex set of $G$ that has maximum weight.
Our main idea relies on the similar sets that we used for the dynamic programming of interval graphs.
That is, we need to describe appropriate sets that define the solutions to be chosen only from a specific part of the considered subproblem.
Although for interval graphs we showed that adding two vertices into such sets is enough,
for permutation graphs we need to consider at most two newly crossing pairs which corresponds to consider at most four newly vertices.
Moreover as a crossing pair may belong to $\mathcal{I}$ we are enforced to describe such a situation into two different sets for each subproblem.
\begin{definition}[$A$-sets]\label{def:sfvs:perm:a}
Let $ij\in\mathcal{X}$. Then,
\begin{displaymath}
A_{ij}=\df\max_{w}\{X\subseteq V_{ij}:G[X]\in\mathcal{F}_{S}\}.
\end{displaymath}
\end{definition}
\begin{definition}[$B$-sets]\label{def:sfvs:perm:b:1}
Let $ij\in\mathcal{X}$ and let $x\in V\setminus V_{ij}$. Then,
\begin{displaymath}
B_{ij}^{xx}=\df\max_{w}\{X\subseteq V_{ij}:G[X\cup\{x\}]\in\mathcal{F}_{S}\}.
\end{displaymath}
\end{definition}
\begin{definition}[$B$-sets]\label{def:sfvs:perm:b:2}
Let $ij\in\mathcal{X}$ and $xy\in\mathcal{X}\setminus\mathcal{I}$ such that $j<_{t}y$, $i<_{b}x$, and $x,y\notin S$. Then,
\begin{displaymath}
B_{ij}^{xy}=\df\max_{w}\{X\subseteq V_{ij}:G[X\cup\{x,y\}]\in\mathcal{F}_{S}\}.
\end{displaymath}
\end{definition}
\begin{definition}[$C$-sets]\label{def:sfvs:perm:c:1}
Let $ij\in\mathcal{X}$, $xy\in\mathcal{X}\setminus\mathcal{I}$, and $z\in V\setminus (V_{ij} \setminus \{x,y\})$
such that $xy<_{\ell}zz$, at least one of $x,y$ is adjacent to $z$, $j<_{t}y$, $i<_{b}x$, and $x,y,z\notin S$. Then,
\begin{displaymath}
C_{ij}^{xy,zz}=\df\max_{w}\{X\subseteq V_{ij}:G[X\cup\{x,y,z\}]\in\mathcal{F}_{S}\}.
\end{displaymath}
\end{definition}
\begin{definition}[$C$-sets]\label{def:sfvs:perm:c:2}
Let $ij\in\mathcal{X}$ and $xy,zw\in\mathcal{X}\setminus\mathcal{I}$ such that $xy<_{\ell}zw$, $\{x,w\},\{y,z\}\in E$, $j<_{t}\{y,w\}$, $i<_{b}\{x,z\}$, and $x,y,z,w\notin S$. Then,
\begin{displaymath}
C_{ij}^{xy,zw}=\df\max_{w}\{X\subseteq V_{ij}:G[X\cup\{x,y,z,w\}]\in\mathcal{F}_{S}\}.
\end{displaymath}
\end{definition}

Observe that, since $V_{00}=\{0\}$ and $w(0)\leq0$, $A_{00}=\varnothing$ and, since $V_{\pi(n)n}=V$, $A_{\pi(n)n}=\max_{w}\{X\subseteq V:G[X]\in\mathcal{F}_{S}\}$.
The following lemmas state how to recursively compute all $A$-sets, $B$-sets, and $C$-sets other than $A_{00}$.
Because every crossing pair $ij$ might be of the form $ii$ we first consider
the sets $A_{ii}$, $B_{ii}^{xx}$, $B_{ii}^{xy}$, $C_{ii}^{xy,zz}$, and $C_{ii}^{xy,zw}$.

\begin{lemma}\label{lem:sfvs:perm:i}
Let $i\in V\setminus\{0\}$. Then $A_{ii}=A_{{<}ii}\cup\{i\}$.
\end{lemma}
\begin{proof}
By Observation~\ref{obs:sfvs:perm}~(1) we have $A_{{<}ii}\cup\{i\}\in {A}_{ii}$.
Also notice that the neighbourhood of $i$ in $G[V_{ii}]$ is $\emptyset$.
Thus no subset of $V_{ii}$ that contains $i$ induces an $S$-cycle in $G$, so that $i\in A_{ii}$.
Therefore $A_{ii}=A_{{<}ii}\cup\{i\}$.
\end{proof}

\begin{lemma}\label{lem:sfvs:perm:ix}
Let $i\in V$ and let $x\in V\setminus V_{ii}$.
\begin{enumerate}
\item If $\{i,x\}\notin E$ then $B_{ii}^{xx}=A_{ii}$.
\item If $\{i,x\}\in E$ then $B_{ii}^{xx}=B_{<ii}^{xx}\cup\{i\}$.
\end{enumerate}
\end{lemma}
\begin{proof}
Assume first that $\{i,x\}\notin E$.
Since $x \in V\setminus V_{ii}$ we know that $i<_{t}x$ or $i<_{b}x$.
Moreover as $\{i,x\}\notin E$ we have $i<_{t}x$ and $i<_{b}x$.
Then $x$ has no neighbour in $G[V_{ii}\cup\{x\}]$.
Thus no subset of $V_{ii}\cup\{x\}$ that contains $x$ induces an $S$-cycle in $G$.
Hence $B_{ii}^{xx}=A_{ii}$ follows.

Next assume that $\{i,x\}\in E$.
Then the neighbourhood of $i$ in $G[V_{ii}\cup\{x\}]$ is $\{x\}$.
This means that no subset of $V_{ii}\cup\{x\}$ that contains $i$ induces an $S$-cycle in $G$, so that $i\in B_{ii}^{xx}$.
By Observation~\ref{obs:sfvs:perm}~(1) it follows that $B_{ii}^{xx}=B_{<ii}^{xx}\cup\{i\}$.
\end{proof}

\begin{lemma}\label{lem:sfvs:perm:ixy}
Let $i\in V$ and let $xy\in\mathcal{X}\setminus\mathcal{I}$ such that $i<_{t}y$, $i<_{b}x$, and $x,y\notin S$.
\begin{enumerate}
\item If $\{i,y\}\notin E$ then $B_{ii}^{xy}=B_{ii}^{xx}$.
\item If $\{i,x\}\notin E$ then $B_{ii}^{xy}=B_{ii}^{yy}$.
\item If $\{i,x\},\{i,y\}\in E$ then $B_{ii}^{xy}=\left\{\begin{array}{l@{}l}
B_{<ii}^{xy}&,\textrm{ if }i\in S\\[\myspace]
B_{<ii}^{xy}\cup\{i\}&,\textrm{ if }i\notin S.
\end{array}\right.$
\end{enumerate}
\end{lemma}
\begin{proof}
By $i<_{t}y$, $i<_{b}x$, and the fact that $xy$ is a crossing pair, we have $\{x,i\} <_{t} y$ and $\{y,i\} <_{b} x$.
Assume first that $i$ is non-adjacent to at least one of $x$ and $y$.
Let $\{i,y\}\notin E$.
Then $\{i,x\}<_{t}y$ and $i<_{b}y<_{b}x$, so that the neighbourhood of $y$ in $G[V_{ii}\cup\{x,y\}]$ is $\{x\}$.
Thus no subset of $V_{ii}\cup\{x,y\}$ that contains $y$ induces an $S$-cycle of $G$ which implies that
$B_{ii}^{xy}=B_{ii}^{xx}$.
Completely symmetric arguments apply if $\{i,x\}\notin E$ showing the second statement.

Next assume that $\{i,x\},\{i,y\}\in E$.
Then $x<_{t}i<_{t}y$ and $y<_{b}i<_{b}x$, so that the neighbourhood of $i$ in $G[X\cup\{x,y\}]$ is $\{x,y\}$.
We distinguish two cases according to whether $i$ belongs to $S$.
Suppose that $i\in S$.
Then $\langle i,x,y\rangle$ is an induced $S$-triangle of $G$, so that $i\notin B_{ii}^{xy}$.
Thus by Observation~\ref{obs:sfvs:perm}~(1), $B_{ii}^{xy}=B_{<ii}^{xy}$ holds if $i\in S$.

Suppose next that $i\notin S$.
We will show that no subset of $V_{ii}\cup\{x,y\}$ that contains $i$ induces an $S$-cycle of $G$.
Recall that $i$ is non-adjacent to any vertex of $V_{ii}$ and the only induced cycles of a permutation graph is either a triangle or a square.
Assume that $\langle v_{1},v_{2},i\rangle$ is an induced $S$-triangle of $G$ where $v_{1},v_{2}\in V_{<ii}\cup\{x,y\}$.
Then $\{v_{1},v_{2}\}=\{x,y\}$ leading to a contradiction, because $i,x,y\notin S$.
So let us assume that $\langle v_{1},v_{2},v_{3},i\rangle$ is an induced $S$-square of $G$ where $v_{1},v_{2},v_{3}\in V_{<ii}\cup\{x,y\}$.
By the fact that $i$ only adjacent to $x$ and $y$ in $G[V_{ii}\cup\{x,y\}]$ we have that $v_1,v_3$ correspond to the vertices $x$ and $y$.
This however leads to a contradiction since $\{x,y\}\in E$ and $\{v_1,v_3\}\notin E$ by the induced $S$-square.
Therefore no subset of $V_{ii}\cup\{x,y\}$ that contains $i$ induces an $S$-cycle of $G$, so that $i\in B_{ii}^{xy}$.
By Observation~\ref{obs:sfvs:perm}~(1) $B_{ii}^{xy}=B_{<ii}^{xy}\cup\{i\}$ holds and this completes the proof.
\end{proof}

\begin{lemma}\label{lem:sfvs:perm:ixyz}
Let $i\in V$, $xy\in\mathcal{X}\setminus\mathcal{I}$, and let $z\in V\setminus (V_{ii}\setminus \{x,y\})$ such that
$xy<_{\ell}zz$, at least one of $x,y$ is adjacent to $z$, $i<_{t}y$, $i<_{b}x$, and $x,y,z\notin S$.
\begin{enumerate}
\item If $\{i,z\}\notin E$ then $C_{ii}^{xy,zz}=B_{ii}^{xy}$.
\item If $\{i,z\}\in E$ then $C_{ii}^{xy,zz}=\left\{\begin{array}{l@{}l}
C_{<ii}^{xy,zz}, &\textrm{ if }i\in S\\[\myspace]
C_{<ii}^{xy,zz}\cup\{i\}, &\textrm{ if }i\notin S.
\end{array}\right.$
\end{enumerate}
\end{lemma}
\begin{proof}
Since $z\in V\setminus (V_{ii}\setminus \{x,y\})$, we have $i <_{t} z$ or $i<_{b} z$.
Assume first that $\{i,z\}\notin E$.
Observe that this means that $i <_{t} z$ and $i<_{b} z$.
Then $z$ is non-adjacent to any vertex of $V_{ii}$ so that the neighborhood of $z$ in $G[V_{ii}\cup\{x,y,z\}]$ is a subset of $\{x,y\}$.
Since $x,y,z \notin S$, no subset of $V_{ii}\cup\{x,y,z\}$ that contains $z$ induces an $S$-cycle in $G$.
Thus $C_{ii}^{xy,zz}=B_{ii}^{xy}$.

Assume next that $\{i,z\}\in E$. This means that either $i <_t z$ and $z <_b i$ hold, or $z <_t i$ and $i <_b z$ hold.
Since $i<_{t}y$ and $i<_{b}x$, we get either $i <_t \{y,z\}$ and $z <_b i <_b x$, or $z <_t i <_t y$ and $i <_b \{x,z\}$.
Moreover since $xy$ is a crossing pair and $xy<_{\ell}zz$, exactly one of following holds:
\begin{itemize}
\item $\{i,x\}<_{t}\{y,z\}$ and $y<_{b}z<_{b}i<_{b}x$;
\item $x<_{t}z<_{t}i<_{t}y$ and $\{i,y\}<_{b}\{x,z\}$.
\end{itemize}
This means that $y,z \in N(i)$ and $x$ is adjacent to $z$, or $x,z \in N(i)$ and $y$ is adjacent to $z$.
We distinguish two cases depending on whether $i$ belongs to $S$.
\begin{itemize}
\item Let $i\in S$. We will show that $i\notin C_{ii}^{xy,zz}$.
If both $x$ and $y$ are adjacent to $i$ then $\langle i,x,y \rangle$ is an induced $S$-triangle in $G$.
Thus either $y,z \in N(i)$ and $x$ is adjacent to $z$, or $x,z \in N(i)$ and $y$ is adjacent to $z$.
Assume the former, that is, $y,z \in N(i)$, $x \notin N(i)$, and $x$ is adjacent to $z$.
If $\{y,z\}\in E$ then $\langle i,y,z\rangle$ is an induced $S$-triangle and
if $\{y,z\}\notin E$ then $\langle i,y,x,z\rangle$ is an induced $S$-square.
Similarly if $x,z \in N(i)$, $y \notin N(i)$, and $y$ is adjacent to $z$ we obtain an induced $S$-cycle in $G$.
Therefore in all cases $i\notin C_{ii}^{xy,zz}$ and by Observation~\ref{obs:sfvs:perm}~(1) we get $C_{ii}^{xy,zz}=C_{<ii}^{xy,zz}$.

\item Let $i\notin S$. We will show that $i \in C_{ii}^{xy,zz}$.
Assume for contradiction that there is an induced $S$-triangle $\langle v_{1},v_{2},i\rangle$ or $S$-square $\langle v_{1},v_{2},v_{3},i\rangle$ in $G$ where $v_{1},v_{2},v_{3}\in V_{<ii}\cup\{x,y,z\}$.
Notice that $N(i) \cap V_{<ii} = \emptyset$ so that $\{v_{1},v_{2}\}\subset\{x,y,z\}$ or $\{v_{1},v_{3}\}\subset\{x,y,z\}$, respectively.
In the former case we reach a contradiction because $i,x,y,z\notin S$.
In the latter case for the same reason notice that $v_{2} \in S$ which implies that $v_{2}\in V_{<ii}$.
If $\{v_{1},v_{3}\}=\{x,y\}$ then we reach a contradiction to the $S$-square $\langle v_{1},v_{2},v_{3},i\rangle$ because $\{x,y\}\in E$. Thus $\{v_{1},v_{3}\}=\{y,z\}$ or $\{v_{1},v_{3}\}=\{x,z\}$.
Without loss of generality assume that $\{v_{1},v_{3}\}=\{y,z\}$.
Then $\{y,z\}\notin E$, for otherwise we reach again a contradiction to the given $S$-square.
Observe that $\{y,z\}\notin E$ implies that $\{x,z\}\in E$ by the hypothesis for $z$.
This however shows that $\langle y,v_{2},x\rangle$ or
$\langle y,v_{2},z,x\rangle$ induce an $S$-triangle or an $S$-square of $G$ without $i$ depending on whether $x$ is adjacent to $v_2$, so that $v_{2}\notin C_{ii}^{xy,zz}$.
Therefore in all cases we reach a contradiction which means that $i\in C_{ii}^{xy,zz}$ and by Observation~\ref{obs:sfvs:perm}~(1), $C_{ii}^{xy,zz}=C_{<ii}^{xy,zz}\cup\{i\}$ holds.
\end{itemize}
In each case we have showed the described equations and this completes the proof.
\end{proof}

\begin{lemma}\label{lem:sfvs:perm:ixyzw}
Let $i\in V$ and let $xy,zw\in\mathcal{X}\setminus\mathcal{I}$ such that $xy<_{\ell}zw$, $\{x,w\},\{y,z\}\in E$, $i<_{t}\{y,w\}$, $i<_{b}\{x,z\}$, and $x,y,z,w\notin S$.
\begin{enumerate}
\item If $\{i,w\}\notin E$ then $C_{ii}^{xy,zw}=C_{ii}^{xy,zz}$.
\item If $\{i,z\}\notin E$ then $C_{ii}^{xy,zw}=C_{ii}^{xy,ww}$.
\item If $\{i,z\},\{i,w\}\in E$ then $C_{ii}^{xy,zw}=\left\{\begin{array}{l@{}l}
C_{<ii}^{xy,zw},& \textrm{ if }i\in S\\[\myspace]
C_{<ii}^{xy,zw}\cup\{i\},& \textrm{ if }i\notin S.
\end{array}\right.$
\end{enumerate}
\end{lemma}
\begin{proof}
Observe that $x,y,z,w \in V \setminus V_{ii}$ because $i<_{t}\{y,w\}$ and $i<_{b}\{x,z\}$.
Assume first that $\{i,w\}\notin E$.
Since $i<_{t}w$ and $i<_{b}w$, $w$ has no neighbour in $V_{ii}$.
Thus the neighbourhood of $w$ in $G[V_{ii}\cup\{x,y,z,w\}]$ is a subset of $\{x,y,z\}$.
We will show that $w \notin {C}_{ii}^{xy,zw}$.
Assume that a subset of $V_{ii}\cup\{x,y,z,w\}$ that contains $w$ induces an $S$-cycle in $G$.
If $\langle v_{1},v_{2},w\rangle$ is an induced $S$-triangle of $G$ then $\{v_{1},v_{2}\}\subset\{x,y,z\}$ which leads to a contradiction, because $x,y,z,w\notin S$.
Suppose that $\langle v_{1},v_{2},v_{3},w\rangle$ is an induced $S$-square of $G$.
Then $\{v_{1},v_{3}\}\subset\{x,y,z\}$ and, since $x,y,z,w\notin S$ we know that $v_2 \in S$ and $v_{2}\in V_{ii}$.
\begin{itemize}
\item Assume that $\{v_{1},v_{3}\}=\{x,y\}$ or $\{v_{1},v_{3}\}=\{y,z\}$.
Then we reach a contradiction to the induced $S$-square, because $\{x,y\},\{y,z\}\in E$.
\item Assume that $\{v_{1},v_{3}\}=\{x,z\}$. If $\{x,z\}\in E$ then $\langle v_{1},v_{2},v_{3},w\rangle$ does not induce an $S$-square. If $\{x,z\}\notin E$ then $\langle x,v_{2},y\rangle$ or $\langle x,v_{2},z,y\rangle$ induce an $S$-triangle or an $S$-square in $G$ which reach to a contradiction to $v_{2}\notin C_{ii}^{xy,zz}$.
\end{itemize}
Therefore, if a subset of $V_{ii}\cup\{x,y,z,w\}$ that contains $w$ induces an $S$-cycle of $G$, then its non-empty intersection with $V_{ii}$ is not a subset of $C_{ii}^{xy,zz}$ which implies that $C_{ii}^{xy,zw}=C_{ii}^{xy,zz}$.
The case for $\{i,z\}\notin E$ is completely symmetric showing the second statement.

Let $\{i,z\},\{i,w\}\in E$.
Then either $i <_t \{z,w\}$ and $\{z,w\}<_b i$, or $\{z,w\} <_t i$ and $i <_b \{z,w\}$.
Since $xy<_{\ell}zw$, $i<_{t}\{y,w\}$, and $i<_{b}\{x,z\}$, the following hold:
\begin{itemize}
\item $x<_{t}z<_{t}i<_{t}\{y,w\}$ and
\item $y<_{b}w<_{b}i<_{b}\{x,z\}$.
\end{itemize}
Thus the neighborhood of $i$ in $G[V_{ii}\cup\{x,y,z,w\}]$ is $\{x,y,z,w\}$.
Assume that $i\in S$. Then $\langle i,x,y\rangle$ is an $S$-triangle of $G$ which implies $i\notin C_{ii}^{xy,zw}$.
By Observation~\ref{obs:sfvs:perm}~(1) we get $C_{ii}^{xy,zw}=C_{<ii}^{xy,zw}$.
Let us assume that $i\notin S$.
We will show that if a subset of $V_{ii}\cup\{x,y,z,w\}$ that contains $i$ induces an $S$-cycle of $G$, then its non-empty intersection with $V_{<ii}$ is not a subset of $C_{ii}^{xy,zw}$.
\begin{itemize}
\item Let $v_{1},v_{2}\in V_{<ii}\cup\{x,y,z,w\}$ such that $\langle v_{1},v_{2},i\rangle$ is an induced $S$-triangle of $G$. Then $\{v_{1},v_{2}\}\subset\{x,y,z,w\}$, a contradiction, because $i,x,y,z,w\notin S$.
\item Let $v_{1},v_{2},v_{3}\in V_{<ii}\cup\{x,y,z,w\}$ such that $\langle v_{1},v_{2},v_{3},i\rangle$ is an induced $S$-square of $G$. Then $\{v_{1},v_{3}\}\subset\{x,y,z,w\}$ and, since $i,x,y,z,w\notin S$, $v_{2}\in S$. Thus $v_{2}\in V_{<ii}$.
Because $v_1,v_3$ are non-adjacent, we have $\{v_{1},v_{3}\}=\{x,z\}$ or $\{v_{1},v_{3}\}=\{y,w\}$.
In both cases we reach a contradiction since $\langle x,v_{2},z,y\rangle$ or $\langle y,v_{2},w,z\rangle$ induce $S$-squares in $G$.
\end{itemize}
Thus if $i\notin S$ then $i\in C_{ii}^{xy,zw}$.
Therefore by Observation~\ref{obs:sfvs:perm}~(1) we obtain $C_{ii}^{xy,zw}=C_{<ii}^{xy,zw}\cup\{i\}$.
\end{proof}

Based on Lemmas~\ref{lem:sfvs:perm:i}--\ref{lem:sfvs:perm:ixyzw}, for each crossing pair of the form $ii$ we can describe its subsolution by using appropriate formulations of the $A$-, $B$-, or $C$-sets.
In the forthcoming lemmas we give the recursive formulations for the sets $A_{ij}$, $B_{ij}^{xx}$, $B_{ij}^{xy}$, $C_{ij}^{xy,zz}$, and $C_{ij}^{xy,zw}$ whenever $ij \in\mathcal{X}\setminus\mathcal{I}$ which particularly means that $i$ and $j$ are distinct vertices in $G$.

\begin{lemma}\label{lem:sfvs:perm:ij}
Let $ij\in\mathcal{X}\setminus\mathcal{I}$. Then,
$$
A_{ij}=\left\{\begin{array}{l@{}l}
\max_{w}\left\{A_{\eqslantless ij},A_{\leqslant ij},B_{\ll jj}^{ii}\cup\{i,j\},B_{\ll ii}^{jj}\cup\{i,j\}\right\}, &
\textrm{ if }i\in S\textrm{ or }j\in S\\[\myspace]
\max_{w}\left\{A_{\eqslantless ij},A_{\leqslant ij},B_{<ij}^{ij}\cup\{i,j\}\right\}, &
\textrm{ if }i,j\notin S.
\end{array}\right.$$
\end{lemma}
\begin{proof}
Let $j\notin A_{ij}$. Then by Observation~\ref{obs:sfvs:perm}~(1) it follows that $A_{ij}=A_{\eqslantless ij}$.
Similarly if $i\notin A_{ij}$ then $A_{ij}=A_{\leqslant ij}$.
For the rest of the proof we assume that $i,j\in A_{ij}$.
Notice that by Observation~\ref{obs:sfvs:perm}~(1)  we have $A_{ij}\setminus\{i,j\}\subseteq V_{<ij}$.
We distinguish two cases according to whether $i$ or $j$ belong to $S$.
\begin{itemize}
\item Assume that $i,j\notin S$. Then $A_{ij}=B_{<ij}^{ij}\cup\{i,j\}$ holds which completes the second description in the formula.
\item Assume that $i\in S$ or $j\in S$.
Let $h\in A_{ij}\setminus\{i,j\}$ such that $\{h,i\},\{h,j\}\in E$.
Then $\langle h,i,j\rangle$ is an induced $S$-triangle in $G$, resulting a contradiction to $i,j\in A_{ij}$.
Thus for every $h\in A_{ij}\setminus\{i,j\}$ we know that $\{h,i\}\notin E$ or $\{h,j\}\notin E$.
Let $g,h\in A_{ij}\setminus\{i,j\}$ such that $\{g,j\},\{h,i\}\in E$ and $\{g,i\},\{h,j\}\notin E$.
Observe that $\{g,h\} <_{b} i$ and $\{g,h\} <_{t} j$.
Since $ij$ is a crossing pair we know that $i <_t j$ and $j <_b i$.
If $i<_{t}g$ or $j<_{b}h$ then $g$ is adjacent to $i$ or $h$ is adjacent to $j$, leading to a contradiction.
%
Thu $g<_{t}i<_{t}h$ and $h<_{b}j<_{b}g$ hold which imply that $\{g,h\}\in E$.
Hence $\langle g,h,i,j\rangle$ is an induced $S$-square in $G$, a contradiction.
This means that all vertices of $A_{ij}\setminus\{i,j\}$ are non-adjacent to $i$ or $j$ or both.
Then by Observation~\ref{obs:sfvs:perm}~(2) it follows that either $A_{ij}\setminus\{i,j\}\subseteq V_{\ll jj}$ or $A_{ij}\setminus\{i,j\}\subseteq V_{\ll ii}$.

Suppose that the former holds, that is $A_{ij}\setminus\{i,j\}\subseteq V_{\ll jj}$.
The neighborhood of $j$ in $G[V_{\ll jj}\cup\{i,j\}]$ is $\{i\}$.
Thus no subset of $V_{\ll jj}\cup\{i,j\}$ that contains $j$ induces an $S$-cycle in $G$.
This means that $A_{ij}=B_{\ll jj}^{ii}\cup\{i,j\}$ as described in the first description in the given formula.
If $A_{ij}\setminus\{i,j\}\subseteq V_{\ll ii}$ then completely symmetric we have $A_{ij}=B_{\ll ii}^{jj}\cup\{i,j\}$.
\end{itemize}
Therefore the corresponding formulas given in the statement follow.
\end{proof}

With the next two lemmas we describe recursively the sets $B_{ij}^{xx}$ and $B_{ij}^{xy}$.
Given a set of vertices $L \subseteq V$ we define the following crossing pair.
The {\it leftmost crossing pair} of $L$ is the crossing pair $x'y' \in \mathcal{X}\setminus\mathcal{I}$ with $x',y' \in L$ such that for any $z'w' \in \mathcal{X}\setminus\mathcal{I}$ with $z',w' \in L$,
$x'y' \leq_{\ell} z'w'$ holds.

\begin{lemma}\label{lem:sfvs:perm:ijx}
Let $ij\in\mathcal{X}\setminus\mathcal{I}$ and let $x\in V\setminus V_{ij}$.
Moreover let $x'y'$ be the leftmost crossing pair of $\{i,j,x\}$ and
let $z'$ be the vertex of $\{i,j,x\}\setminus \{x',y'\}$.
\begin{enumerate}
\item If $\{i,x\},\{j,x\}\notin E$ then $B_{ij}^{xx}=A_{ij}$.
\item If $\{i,x\}\in E$ and $\{j,x\}\notin E$ then
\begin{displaymath}
B_{ij}^{xx}=\left\{\begin{array}{l@{}l}
\max_{w}\left\{B_{\eqslantless ij}^{xx},B_{\leqslant ij}^{xx},B_{\ll jj}^{ii}\cup\{i,j\},B_{\ll ix}^{jj}\cup\{i,j\}\right\}, & \textrm{ if }i\in S\textrm{ or }j\in S\\[\myspace]
\max_{w}\left\{B_{\eqslantless ij}^{xx},B_{\leqslant ij}^{xx},B_{<ij\ll xx}^{ij}\cup\{i,j\}\right\}, & \textrm{ if }i,j\notin S\textrm{ and }x\in S\\[\myspace]
\max_{w}\left\{B_{\eqslantless ij}^{xx},B_{\leqslant ij}^{xx},C_{<ij}^{x'y',z'z'}\cup\{i,j\}\right\}, & \textrm{ if }i,j,x\notin S.
\end{array}\right.
\end{displaymath}
\item If $\{i,x\}\notin E$ and $\{j,x\}\in E$ then
\begin{displaymath}
B_{ij}^{xx}=\left\{\begin{array}{l@{}l}
\max_{w}\left\{B_{\eqslantless ij}^{xx},B_{\leqslant ij}^{xx},B_{\ll xj}^{ii}\cup\{i,j\},B_{\ll ii}^{jj}\cup\{i,j\}\right\}, & \textrm{ if }i\in S\textrm{ or }j\in S\\[\myspace]
\max_{w}\left\{B_{\eqslantless ij}^{xx},B_{\leqslant ij}^{xx},B_{<ij\ll xx}^{ij}\cup\{i,j\}\right\}, & \textrm{ if }i,j\notin S\textrm{ and }x\in S\\[\myspace]
\max_{w}\left\{B_{\eqslantless ij}^{xx},B_{\leqslant ij}^{xx},C_{<ij}^{x'y',z'z'}\cup\{i,j\}\right\}, & \textrm{ if }i,j,x\notin S.
\end{array}\right.
\end{displaymath}
\item If $\{i,x\},\{j,x\}\in E$ then
\begin{displaymath}
B_{ij}^{xx}=\left\{\begin{array}{l@{}l}
\max_{w}\left\{B_{\eqslantless ij}^{xx},B_{\leqslant ij}^{xx}\right\}, & \textrm{ if }i\in S\textrm{ or }j\in S\textrm{ or }x\in S\\[\myspace]
\max_{w}\left\{B_{\eqslantless ij}^{xx},B_{\leqslant ij}^{xx},C_{<ij}^{x'y',z'z'}\cup\{i,j\}\right\}, & \textrm{ if }i,j,x\notin S.
\end{array}\right.
\end{displaymath}
\end{enumerate}
\end{lemma}
\begin{proof}
Let us assume first that $\{i,x\},\{j,x\}\notin E$.
Since $i<_{t}j$, $j<_{b}i$, and $x \in V\setminus V_{ij}$, we know that $i<_{t}j<_{t}x$ and $j<_{b}i<_{b}x$.
Thus the neighborhood of $x$ in $G[V_{ij}\cup\{x\}]$ is $\emptyset$.
Hence no subset of $V_{ij}\cup\{x\}$ that contains $x$ induces an $S$-cycle of $G$ and it follows that $B_{ij}^{xx}=A_{ij}$ as described in the first statement.

Assume next that $\{i,x\}\in E$ or $\{j,x\}\in E$.
Let $j\notin B_{ij}^{xx}$. By Observation~\ref{obs:sfvs:perm}~(1) we get $B_{ij}^{xx}=B_{\eqslantless ij}^{xx}$.
Similarly, if $i\notin B_{ij}^{xx}$ then $B_{ij}^{xx}=B_{\leqslant ij}^{xx}$.
So suppose next that $i,j\in B_{ij}^{xx}$.
Notice that $B_{ij}^{xx}\setminus\{i,j\}\subseteq V_{<ij}$ by Observation~\ref{obs:sfvs:perm}~(1).
We distinguish the following cases.
\begin{itemize}
\item Assume that $\{i,x\}\in E$ and $\{j,x\}\notin E$.
Since $x \notin V_{ij}$, $j <_t x$ or $i <_b x$. If $i <_b x$ then $x<_t i$ as $\{i,x\}\in E$ but then $x <_t j$ and $j <_b i <_b x$ so that $\{j,x\}\in E$, leading to a contradiction. Thus $j <_t x$ holds.
Since $\{i,x\}\in E$ and $\{j,x\}\notin E$, we have
$j <_b < x <_b i$ and $i <_t < j <_t x$.
We further reduce to the situations depending on whether $i,j,x$ belong to $S$.
\begin{itemize}
\item Let $i\in S$ or $j\in S$.
Let $h\in B_{ij}^{xx}\setminus\{i,j\}$ such that $\{h,i\},\{h,j\}\in E$. Then $\langle h,i,j\rangle$ is an induced $S$-triangle in $G$, a contradiction.
So $\{h,i\}\notin E$ or $\{h,j\}\notin E$ for every $h\in B_{ij}^{xx}\setminus\{i,j\}$.
Let $g,h\in B_{ij}^{xx}\setminus\{i,j\}$ such that $\{g,j\},\{h,i\}\in E$.
Since $\{g,h\}<_{b}i$ and $\{g,h\}<_{t}j$ by the choice of $g,h \in B_{ij}$, it follows that $g<_{t}i<_{t}h$ and $h<_{b}j<_{b}g$.
Thus $\{g,h\}\in E$.
This however results in an induced $S$-square $\langle g,h,i,j\rangle$ in $G$.
This means that for every $h\in B_{ij}^{xx}\setminus\{i,j\}$ either $\{h,i\}\notin E$ or $\{h,j\}\notin E$.
By Observation~\ref{obs:sfvs:perm}~(2) it follows that either $B_{ij}^{xx}\setminus\{i,j\}\subseteq V_{\ll jj}$ or $B_{ij}^{xx}\setminus\{i,j\}\subseteq V_{\ll ii}$.

In the former case notice that both $j$ and $x$ in $G[V_{\ll jj}\cup\{i,j,x\}]$ are adjacent only to $i$.
Thus no subset of $V_{\ll jj}\cup\{i,j,x\}$ that contains $j$ or $x$ induces an $S$-cycle of $G$ so that $B_{ij}^{xx}=B_{\ll jj}^{ii}\cup\{i,j\}$ as described.

In the latter case we have $B_{ij}^{xx}\setminus\{i,j\}\subseteq V_{\ll ii}$.
Let $h\in B_{ij}^{xx}\setminus\{i,j\}$. We show that $\{h,x\}\notin E$.
Assume for contradiction that $\{h,x\}\in E$.
This means that either $h <_t x$ and $x <_b h$, or $x <_t h$ and $h <_b x$.
Observe that $h <_t j$ and $h <_b i$.
Since $j <_b < x <_b i$ and $i <_t < j <_t x$, we get the following:
\begin{itemize}
\item $h<_{t}i<_{t}j<_{t}x$ and
\item $j<_{b}x<_{b}h<_{b}i$.
\end{itemize}
Thus $\{h,j\}\in E$.
This however shows that $\langle h,j,i,x\rangle$ is an induced $S$-square in $G$, leading to a contradiction.
Thus $\{h,x\}\notin E$ for every $h\in B_{ij}^{xx}\setminus\{i,j\}$.
Then by Observation \ref{obs:sfvs:perm}~(3) it follows that $B_{ij}^{xx}\setminus\{i,j\}\subseteq V_{\ll ix}$.
This means that $i$ and $x$ are only adjacent to $j$ in $G[V_{\ll ix}\cup\{i,j,x\}]$.
Hence no subset of $V_{\ll ix}\cup\{i,j,x\}$ that contains $i$ or $x$ induces an $S$-cycle in $G$, so that $B_{ij}^{xx}=B_{\ll ix}^{jj}\cup\{i,j\}$ as described.

\item Let $i,j\notin S$ and $x\in S$. Let $h\in B_{ij}^{xx}\setminus\{i,j\}$.
We show that $\{h,x\}\notin E$. Assume for contradiction that $\{h,x\}\in E$.
Then either $h <_t x$ and $x <_b h$ hold, or $x <_t h$ and $h <_b x$ hold.
Since $\{i,x\}\in E$, $\{j,x\}\notin E$, and $ij$ is a crossing pair, we have
\begin{itemize}
\item $\{h,i\}<_{t}j<_{t}x$ and
\item $j<_{b}x<_{b}h<_{b}i$
\end{itemize}
implying that $\{h,j\}\in E$.
If $\{h,i\}\in E$ then $\langle h,i,j\rangle$ is an induced $S$-triangle whereas if $\{h,i\}\notin E$ then $\langle h,j,i,x\rangle$ is an induced $S$-square.
Thus we reach a contradiction so that $\{h,x\}\notin E$. Then by Observation~\ref{obs:sfvs:perm}~(5) we get $B_{ij}^{xx}=B_{<ij\ll xx}^{ij}\cup\{i,j\}$, as described.

\item Let $i,j,x\notin S$. By the fact $B_{ij}^{xx}\setminus\{i,j\}\subseteq V_{<ij}$, we have $B_{ij}^{xx}=B_{<ij}^{x'y',z'z'}\cup\{i,j\}$.
\end{itemize}
\item Assume that $\{i,x\}\notin E$ and $\{j,x\}\in E$. This case is symmetric to the one above, so that the following hold:
\begin{itemize}
\item If $i\in S$ or $j\in S$ then either $B_{ij}^{xx}=B_{\ll xj}^{ii}\cup\{i,j\}$ or $B_{ij}^{xx}=B_{\ll ii}^{jj}\cup\{i,j\}$.
\item If $x\in S$ then $B_{ij}^{xx}=B_{<ij\ll xx}^{ij}\cup\{i,j\}$.
\item If $i,j,x\notin S$ then $B_{ij}^{xx}=B_{<ij}^{x'y',z'z'}$.
\end{itemize}
\item Assume that both $\{i,x\}, \{j,x\}\in E$.
Then no vertex of $\{i,x,y\}$ can belong to $S$ as $\langle i,j,x\rangle$ is an induced triangle in $G$.
Since $B_{ij}^{xx}\setminus\{i,j\}\subseteq V_{<ij}$, we get $B_{ij}^{xx}=B_{<ij}^{x'y',z'z'}\cup\{i,j\}$.
\end{itemize}
Therefore every case results in the described statement of the formulas as required.
\end{proof}

Let $ij,xy\in\mathcal{X}\setminus\mathcal{I}$ such that $j<_{t}y$ and $i<_{b}x$.
It is not difficult to see that if we remove any crossing pair $uv$ from $\{i,j,x,y\}$ then the remaining set contains exactly two vertices that are adjacent.

\begin{lemma}\label{lem:sfvs:perm:ijxy}
Let $ij,xy\in\mathcal{X}\setminus\mathcal{I}$ such that $j<_{t}y$, $i<_{b}x$ and $x,y\notin S$.
Moreover let $x'y'$ be the leftmost crossing pair of $\{i,j,x,y\}$ and
let $z'w'$ be the crossing pair of $\{i,j,x,y\}\setminus\{x',y'\}$.
\begin{enumerate}
\item If $\{i,y\}\notin E$ then $B_{ij}^{xy}=B_{ij}^{xx}$.
\item If $\{j,x\}\notin E$ then $B_{ij}^{xy}=B_{ij}^{yy}$.
\item If $\{i,y\},\{j,x\}\in E$ then
\begin{displaymath}
B_{ij}^{xy}=\left\{\begin{array}{l@{}l}
\max_{w}\left\{B_{\eqslantless ij}^{xy},B_{\leqslant ij}^{xy}\right\}, & \textrm{ if }i\in S\textrm{ or }j\in S\\[\myspace]
\max_{w}\left\{B_{\eqslantless ij}^{xy},B_{\leqslant ij}^{xy},C_{<ij}^{x'y',z'w'}\cup\{i,j\}\right\}, & \textrm{ if }i,j\notin S.
\end{array}\right.
\end{displaymath}
\end{enumerate}
\end{lemma}
\begin{proof}
Assume that $\{i,y\}\notin E$.
Then $i <_b y$ since $i <_t x <_t y$. Thus
\begin{itemize}
\item $i<_{t}j$, $\{j,x\}<_{t}y$, and
\item $j<_{b}i<_{b}y<_{b}x$,
\end{itemize}
so that the neighborhood of $y$ in $G[V_{ij}\cup\{x,y\}]$ is $\{x\}$.
Thus no subset of $V_{ij}\cup\{x,y\}$ that contains $y$ induces an $S$-cycle in $G$.
Therefore $B_{ij}^{xy}=B_{ij}^{xx}$ as described.
If $\{j,x\}\notin E$ then $i$ is non-adjacent to $x$ and similar to the previous case we obtain $B_{ij}^{xy}=B_{ij}^{yy}$.

Assume that $\{i,y\},\{j,x\}\in E$.
We distinguish cases depending on whether $i$ or $j$ belong to the solution.
Assume first that at least one of $i$ or $j$ does not belong to $B_{ij}^{xy}$.
Let $j\notin B_{ij}^{xy}$. By Observation~\ref{obs:sfvs:perm}~(1) we have $B_{ij}^{xy}=B_{\eqslantless ij}^{xy}$.
If $i\notin B_{ij}^{xy}$ then in a similar fashion we get $B_{ij}^{xy}=B_{\leqslant ij}^{xy}$.

Next assume that $i,j\in B_{ij}^{xy}$.
Notice that by Observation~\ref{obs:sfvs:perm}~(1), we have $B_{ij}^{xy}\setminus\{i,j\}\subseteq V_{<ij}$.
Let us show that both $i$ and $j$ do not belong to $S$.
If $\{i,x\}\in E$ or $\{j,y\}\in E$ then $\langle i,x,y\rangle$ or $\langle j,x,y\rangle$ induce a triangle in $G$, since $\{i,y\},\{j,x\}\in E$.
Otherwise, $\{i,x\},\{j,y\}\notin E$, so that $\langle i,j,x,y\rangle$ is an induced square in $G$.
Thus in any case there is an $S$-cycle in $G$ whenever $i \in S$ or $j\in S$ which lead to a contradiction to the fact $i,j\in B_{ij}^{xy}$.
Hence $i,j\notin S$.
Since $B_{ij}^{xy}\setminus\{i,j\}\subseteq V_{<ij}$, it follows $B_{ij}^{xy}=C_{<ij}^{x'y',z'w'}\cup\{i,j\}$ as required.
\end{proof}

With the following two lemmas we consider the last two cases that correspond to the sets $C_{ij}^{xy,zz}$ and $C_{ij}^{xy,zw}$, respectively.

\begin{lemma}\label{lem:sfvs:perm:ijxyz}
Let $ij,xy\in\mathcal{X}\setminus\mathcal{I}$ and let $z\in V\setminus V_{ij}$ such that $xy<_{\ell}zz$,
at least one of $x,y$ is adjacent to $z$, $j<_{t}y$, $i<_{b}x$, and $x,y,z\notin S$.
Moreover let $x'y'$ be the leftmost crossing pair of $\{i,j,x,y,z\}$ and
let $z'w'$ be the leftmost crossing pair of $\{i,j,x,y,z\}\setminus\{x',y'\}$.
\begin{enumerate}
\item If $\{i,z\},\{j,z\}\notin E$ then $C_{ij}^{xy,zz}=B_{ij}^{xy}$.
\item If $\{i,z\}\in E$ or $\{j,z\}\in E$ then
\begin{displaymath}
C_{ij}^{xy,zz}=\left\{\begin{array}{l@{}l}
\max_{w}\left\{C_{\eqslantless ij}^{xy,zz},C_{\leqslant ij}^{xy,zz}\right\}, & \textrm{ if }i\in S\textrm{ or }j\in S\\[\myspace]
\max_{w}\left\{C_{\eqslantless ij}^{xy,zz},C_{\leqslant ij}^{xy,zz},C_{<ij}^{x'y',z'w'}\cup\{i,j\}\right\}, & \textrm{ if }i,j\notin S.
\end{array}\right.
\end{displaymath}
\end{enumerate}
\end{lemma}
\begin{proof}
Assume first that $\{i,z\},\{j,z\}\notin E$. Then $i<_{t}j$, $\{j,x\}<_{t}\{y,z\}$, $\{i,y\}<_{b}\{x,z\}$, and $j<_{b}i$.
This means that the neighborhood of $z$ in $G[V_{ij}\cup\{x,y,z\}]$ is a subset of $\{x,y\}$.
We will show that no subset of $V_{ij}\cup\{x,y,z\}$ that contains $z$ induces an $S$-cycle of $G$.
\begin{itemize}
\item Let $\langle v_{1},v_{2},z\rangle$ be an induced $S$-triangle such that $v_{1},v_{2}\in V_{ij}\cup\{x,y\}$.
Then $\{v_{1},v_{2}\}=\{x,y\}$ which leads to a contradiction, because $x,y,z\notin S$.

\item Let $\langle v_{1},v_{2},v_{3},z\rangle$ be an induced $S$-square such that $v_{1},v_{2},v_{3}\in V_{ij}\cup\{x,y\}$.
Then $\{v_{1},v_{3}\}=\{x,y\}$ which leads to a contradiction, because $\{x,y\}\in E$.
\end{itemize}
Thus no subset of $V_{ij}\cup\{x,y,z\}$ that contains $z$ induces an $S$-cycle of $G$.
Therefore $C_{ij}^{xy,zz}=B_{ij}^{xy}$ holds.

Assume that $\{i,z\}\in E$ or $\{j,z\}\in E$. We distinguish cases depending on whether $i$ or $j$ belong to $C_{ij}^{xy,zz}$.
If $j\notin C_{ij}^{xy,zz}$ or $i\notin C_{ij}^{xy,zz}$ then by Observation~\ref{obs:sfvs:perm}~(1) we get $C_{ij}^{xy,zz}=C_{\eqslantless ij}^{xy,zz}$ or $C_{ij}^{xy,zz}=C_{\leqslant ij}^{xy,zz}$, respectively.
The remaining case is $i,j\in C_{ij}^{xy,zz}$.
Here we will show the described formula given in the second statement.
By Observation~\ref{obs:sfvs:perm}~(3), notice that $C_{ij}^{xy,zz}\setminus\{i,j\}\subseteq V_{<ij}$.

\medskip

\noindent\textrm{Case 1:} \ Assume that $i\in S$ or $j\in S$.
We will show that there is an $S$-cycle that contains $i$ or $j$ leading to a contradiction to the assumption $i,j\in C_{ij}^{xy,zz}$.
Let us assume that $\{i,z\}\in E$; the case for $\{j,z\}\in E$ is completely symmetric.
Thus $i<_t z$ and $z <_b i$ hold or $z<_t i$ and $i <_b z$ hold.
Moreover we know that $x <_{t} z$ and $y <_b z$ because $xy <_{\ell} zz$.
Since $ij$, $xy$ are crossing pairs and $i <_{t} j<_{t}y$, $i<_{b}x$, exactly one of the following holds:
\begin{itemize}
\item $x <_{t} z <_{t} < i <_{t} y$ and $\{i,y\} <_{b} \{x,z\}$;
\item $\{i,x\}<_{t} \{y,z\}$ and $y<_{b}z<_{b} i <_{b} x$.
\end{itemize}
If the former inequalities hold then it is not difficult to see that $\{i,x\},\{y,z\}\in E$.
And if the latter inequalities hold then $\{i,y\},\{x,z\}\in E$.
Suppose that $\{i,x\},\{y,z\}\in E$.
\begin{itemize}
\item If $\{i,y\}\in E$ then $\langle i,x,y\rangle$ is an induced $S$-triangle.
\item If $\{x,z\}\in E$ then $\langle i,x,z\rangle$ is an induced $S$-triangle.
\item If $\{i,y\},\{x,z\}\notin E$ then $\langle i,x,y,z\rangle$ is an induced $S$-square.
\end{itemize}
Next suppose that $\{i,y\},\{x,z\}\in E$.
\begin{itemize}
\item If $\{i,x\}\in E$ then $\langle i,x,y\rangle$ is an induced $S$-triangle.
\item If $\{y,z\}\in E$ then $\langle i,y,z\rangle$ is an induced $S$-triangle.
\item If $\{i,x\},\{y,z\}\notin E$ then $\langle i,y,x,z\rangle$ is an induced $S$-square.
\end{itemize}
Therefore if $i\in S$ or $j\in S$ then $i,j\notin C_{ij}^{xy,zz}$ so that $C_{ij}^{xy,zz}$ can be expressed as $C_{\eqslantless ij}^{xy,zz}$ or $C_{\leqslant ij}^{xy,zz}$, as already explained previously.

\medskip

\noindent\textrm{Case 2:} \ Assume that $i,j\notin S$.
Let $a'$ be the vertex of $\{i,j,x,y,z\}\setminus\{x',y',z',w'\}$.
Observe that $a' \notin S$ since $S \cap \{i,j,x,y,z\} = \emptyset$.
We will show that if a subset of $V_{<ij}\cup\{x',y',z',w',a'\}$ that contains $a'$ induces an $S$-cycle of $G$, then its non-empty intersection with $V_{<ij}$ is not a subset of $C_{<ij}^{x'y',z'w'}$.
Assume for contradiction that a subset of vertices of an induced $S$-cycle that contains $a'$ belongs to $C_{<ij}^{x'y',z'w'}$.
Since the only induced cycles in a permutation graph are triangles or squares we assume that $a'$ is contained in an $S$-triangle or an $S$-square.
\begin{itemize}
\item Let $\langle v_{1},v_{2},a'\rangle$ be an induced $S$-triangle such that $v_{1},v_{2}\in V_{<ij}\cup\{x',y',z',w'\}$.
Since $x',y',z',w'\notin S$, without loss of generality, assume that $v_{1}\in S$ which implies that $v_{1}\in V_{<ij}$.
This means that $v_1 <_{t} j \leq_t y'$ and $v_1 <_b i \leq_{b} x'$.
By the choices of $x'y'$, $z'w'$, and $a'$ we know that $x' <_t z' <_t a'$ and $y' <_b w' <_b a'$.
Since $\{v_{1},a'\} \in E$, $a' <_t v_1$ and $v_1 <_b a'$ hold or $v_1 <_t a'$ and $a' <_b v_1$ hold.
Thus exactly one of the following holds:
\begin{itemize}
\item $x'<_{t}z'<_{t}<a'<_{t}v_{1}<_{t}y'$ and $\{v_{1},y'\}<_{b}\{x',z',a'\}$;
\item $\{v_{1},x'\}<_{t}\{y',w',a'\}$ and $y'<_{b}w'<_{b}<a'<_{b}v_{1}<_{b}x'$.
\end{itemize}
If the former inequalities hold then it is not difficult to see that $\{v_{1},x'\},\{v_{1},z'\}\in E$.
And if the latter inequalities hold then $\{v_{1},y'\},\{v_{1},w'\}\in E$.
Suppose that $\{v_{1},x'\},\{v_{1},z'\}\in E$.
\begin{itemize}
\item If $\{v_{1},y'\}\in E$ then $\langle v_{1},x',y'\rangle$ is an induced $S$-triangle.
\item If $\{x',z'\}\in E$ then $\langle v_{1},x',z'\rangle$ is an induced $S$-triangle.
\item If $\{v_{1},y'\},\{x',z'\}\notin E$ then $\langle v_{1},x',y',z'\rangle$ is induced an $S$-square.
\end{itemize}
Next suppose that $\{v_{1},y'\},\{v_{1},w'\}\in E$.
\begin{itemize}
\item If $\{v_{1},x'\}\in E$ then $\langle v_{1},x',y'\rangle$ is an induced $S$-triangle.
\item If $\{y',w'\}\in E$ then $\langle v_{1},y',w'\rangle$ is an induced $S$-triangle.
\item If $\{v_{1},x'\},\{y',w'\}\notin E$ then $\langle v_{1},y',x',w'\rangle$ is induced an $S$-square.
\end{itemize}
Therefore in all cases we obtain that $v_{1}\notin C_{<ij}^{x'y',z'w'}$.

\item Let $\langle v_{1},v_{2},v_{3},a'\rangle$ be an induced $S$-square such that $v_{1},v_{2},v_{3}\in V_{<ij}\cup\{x',y',z',w'\}$.
If $v_{1}\in S$ or $v_{3}\in S$ then by the previous argument the $S$-vertex is not an element of $C_{<ij}^{x'y',z'w'}$.
So assume that $v_{2}\in S$.
Since $x',y',z',w'\notin S$, we have $v_{2}\in V_{<ij}$, so that $v_2 <_{t} j \leq_t y'$ and $v_2 <_b i \leq_{b} x'$.
By the choices of $x'y'$ and $a'$ we know that $x' <_t a'$ and $y' <_b a'$.
Moreover the induced $S$-square imply that either $\{v_2,a'\} <_t \{v_{1},v_{3}\}$ and $\{v_{1},v_{3}\} <_b \{v_2,a'\}$ hold or $\{v_{1},v_{3}\} <_t \{v_2,a'\}$ and $\{v_2,a'\} <_b \{v_{1},v_{3}\}$ hold.
Thus exactly one of the following holds:
\begin{itemize}
\item $\{v_{2},x'\}<_{t}a'<_{t}\{v_{1},v_{3}\}$ and $\{v_{1},v_{3}\}<_{b}v_{2}<_{b}\{x',a'\}$;
\item $\{v_{1},v_{3}\}<_{t}v_2<_{t}\{y',a'\}$ and $\{v_2,y'\}<_{b}a'<_{b}\{v_{1},v_{3}\}$.
\end{itemize}
The first inequalities imply that $\{v_{1},x'\},\{v_{3},x'\}\in E$ whereas the second inequalities imply that $\{v_{1},y'\},\{v_{3},y'\}\in E$.
Suppose that $\{v_{1},x'\},\{v_{3},x'\}\in E$.
\begin{itemize}
\item If $\{v_{2},x'\}\in E$ then $\langle v_{1},v_{2},x'\rangle$ is an induced $S$-triangle.
\item If $\{v_{2},x'\}\notin E$ then $\langle v_{1},v_{2},v_{3},x'\rangle$ is an induced $S$-square.
\end{itemize}
Similarly if $\{v_{1},y'\},\{v_{3},y'\}\in E$ then we obtain an $S$-triangle $\langle v_{1},v_{2},x'\rangle$ or an $S$-square $\langle v_{1},v_{2},v_{3},x'\rangle$.
Thus $\{v_{1},v_{2},v_{3}\}\cap V_{<ij}$ is not a subset of $C_{<ij}^{x'y',z'w'}$.
\end{itemize}
Therefore if a subset of $V_{<ij}\cup\{x',y',z',w',a'\}$ that contains $a'$ induces an $S$-cycle then its non-empty intersection with $V_{<ij}$ is not a subset of $C_{<ij}^{x'y',z'w'}$.
This particularly implies that $C_{ij}^{xy,zz}=C_{<ij}^{x'y',z'w'}\cup\{i,j\}$ as described in the second statement.
\end{proof}

The next lemma shows how to recursively compute $C_{ij}^{xy,zw}$.
Note that in each case we describe $C_{ij}^{xy,zw}$ as a predefined smaller set of a subsolution that is either in the same form or
has already been described in one of the previous lemmas.

\begin{lemma}\label{lem:sfvs:perm:ijxyzw}
Let $ij,xy,zw\in\mathcal{X}\setminus\mathcal{I}$ such that $xy<_{\ell}zw$, $\{x,w\},\{y,z\}\in E$, $j<_{t}\{y,w\}$, $i<_{b}\{x,z\}$, and $x,y,z,w\notin S$.
Moreover let $x'y'$ be the leftmost crossing pair of $\{i,j,x,y,z,w\}$ and
let $z'w'$ be the leftmost crossing pair of $\{i,j,x,y,z,w\}\setminus\{x',y'\}$.
\begin{enumerate}
\item If $\{i,w\}\notin E$ then $C_{ij}^{xy,zw}=C_{ij}^{xy,zz}$.
\item If $\{j,z\}\notin E$ then $C_{ij}^{xy,zw}=C_{ij}^{xy,ww}$.
\item If $\{i,w\},\{j,z\}\in E$ then
\begin{displaymath}
C_{ij}^{xy,zw}=\left\{\begin{array}{l@{}l}
\max_{w}\left\{C_{\eqslantless ij}^{xy,zw},C_{\leqslant ij}^{xy,zw}\right\}, & \textrm{ if }i\in S\textrm{ or }j\in S\\[\myspace]
\max_{w}\left\{C_{\eqslantless ij}^{xy,zw},C_{\leqslant ij}^{xy,zw},C_{<ij}^{x'y',z'w'}\cup\{i,j\}\right\}, & \textrm{ if }i,j\notin S.
\end{array}\right.
\end{displaymath}
\end{enumerate}
\end{lemma}
\begin{proof}
Assume first that $\{i,w\}\notin E$. Then
$i<_t w$ and $i<_b w$ hold, since $i <_t j <_t w$.
Thus the following hold, as $ij,xy,zw$ are crossing pairs:
\begin{itemize}
\item $i<_{t}j<_{t}\{y,w\}$, $x<_{t}\{y,z\}$, $z<_{t}w$,
\item $j<_{b}i<_{b}w<_{b}z$, and $\{i,y\}<_{b} \{x,w\}$.
\end{itemize}
Then the neighborhood of $w$ in $G[V_{ij}\cup\{x,y,z,w\}]$ is a subset of $\{x,y,z\}$.
We will show that if a subset of $V_{ij}\cup\{x,y,z,w\}$ that contains $w$ induces an $S$-cycle then its non-empty intersection with $V_{ij}$ is not a subset of $C_{ij}^{xy,zz}$.
Since an $S$-cycle in $G$ is only an $S$-triangle or an $S$-square, we distinguish the following two cases:
\begin{itemize}
\item Let $\langle v_{1},v_{2},w\rangle$ be an induced $S$-triangle such that $v_{1},v_{2}\in V_{ij}\cup\{x,y,z\}$.
Then $\{v_{1},v_{2}\}\subset\{x,y,z\}$ since $N(w) \subseteq \{x,y,z\}$, which leads to a contradiction, because $x,y,z,w\notin S$.
\item Let $\langle v_{1},v_{2},v_{3},w\rangle$ be an induced $S$-square such that $v_{1},v_{2},v_{3}\in V_{ij}\cup\{x,y,z\}$.
Then $\{v_{1},v_{3}\}\subset\{x,y,z\}$. Since $x,y,z,w\notin S$, we have $v_{2}\in S$ so that $v_{2}\in V_{ij}$.
\begin{itemize}
\item Assume that $\{v_{1},v_{3}\}=\{x,y\}$ or $\{y,z\}$. Then we reach a contradiction because $\{x,y\},\{y,z\}\in E$.
\item Assume that $\{v_{1},v_{3}\}=\{x,z\}$. Then it is clear that $\{x,z\}\notin E$. This however shows that $\langle x,v_{2},z,y\rangle$ is an induced $S$-square, so that $v_{2}\notin C_{ij}^{xy,zz}$.
\end{itemize}
\end{itemize}
Therefore, if a subset of $V_{ij}\cup\{x,y,z,w\}$ that contains $w$ induces an $S$-cycle then its non-empty intersection with $V_{ij}$ is not a subset of $C_{ij}^{xy,zz}$.
Thus $C_{ij}^{xy,zw}=C_{ij}^{xy,zz}$ holds.

If we assume that $\{j,z\}\notin E$ then similar arguments with the previous case for $\{i,w\}\notin E$ show that $C_{ij}^{xy,zw}=C_{ij}^{xy,ww}$.

Our remaining case is $\{i,w\},\{j,z\}\in E$.
If $j\notin C_{ij}^{xy,zw}$ then by Observation~\ref{obs:sfvs:perm}~(1) we get $C_{ij}^{xy,zw}=C_{\eqslantless ij}^{xy,zw}$.
Similarly if $i\notin C_{ij}^{xy,zw}$ then $C_{ij}^{xy,zw}=C_{\leqslant ij}^{xy,zw}$.
So let us assume that both $i,j$ belong to the solution $C_{ij}^{xy,zw}$, that is $i,j\in C_{ij}^{xy,zw}$.
Notice that by Observation~\ref{obs:sfvs:perm}~(1) we know that $C_{ij}^{xy,zw}\setminus\{i,j\}\subseteq V_{<ij}$.
We distinguish two cases depending on whether $i,j$ belong to $S$.
If $i\in S$ or $j\in S$ then the following induced $S$-cycles show that we reach a contradiction to $i,j\in C_{ij}^{xy,zw}$:
\begin{itemize}
\item If $\{i,z\}\in E$ then $\langle i,j,z\rangle$ is an induced $S$-triangle.
\item If $\{j,w\}\in E$ then $\langle i,j,w\rangle$ is an induced $S$-triangle.
\item If $\{i,z\},\{j,w\}\notin E$, then $\langle i,j,z,w\rangle$ is an induced $S$-square.
\end{itemize}
Thus if $i\in S$ or $j\in S$ then we know that $j\notin C_{ij}^{xy,zw}$ or $i\notin C_{ij}^{xy,zw}$ which shows the first description of $C_{ij}^{xy,zw}$ in the third statement.

Let $i,j\notin S$ and recall that $i,j\in C_{ij}^{xy,zw}$ and $\{i,w\},\{j,z\}\in E$.
Observe that the set $\{i,j,x,y,z,w\}\setminus\{x',y',z',w'\}$ contains two adjacent vertices. 
Let $a'b'$ be the crossing pair of $\{i,j,x,y,z,w\}\setminus\{x',y',z',w'\}$ so that $a' <_t b'$ and $b' <_b a'$.
We will show that if a subset of $V_{<ij}\cup\{x',y',z',w',a',b'\}$ that contains $a'$ or $b'$ induces an $S$-cycle then its non-empty intersection with $V_{<ij}$ is not a subset of $C_{<ij}^{x'y',z'w'}$. Such an $S$-cycle is either an $S$-triangle or an $S$-square.
\begin{itemize}
\item Let $\langle v_{1},v_{2},b'\rangle$ be an induced $S$-triangle such that $v_{1},v_{2}\in V_{<ij} \cup\{x',y',z',w',a'\}$.
Since $x',y',z',w',a',b'\notin S$, without loss of generality, assume that $v_{1}\in S$ so that $v_{1}\in V_{<ij}$.
Then $v_1 <_t j <_t \{y,w\}$ and $v_1 <_b i <_b \{x,z\}$.
Our goal is to show that $v_{1}\notin C_{<ij}^{x'y',z'w'}$.
Since $\{v_1,b'\} \in E$ and $x'<_t a' <_t b'$, we get $v_1 <_t b'$ and $b' <_b v_1$.
Also notice that $x'y'<_{\ell} z'w' <_{\ell} a'b'$ so that $y'<_{b}w'<_{b}b'$.
Thus the following hold:
\begin{itemize}
\item $\{v_{1},x'\}<_{t} \{y',w',b'\}$ and
\item $y'<_{b}w'<_{b}b'<_{b}\{v_{1},x'\}$.
\end{itemize}
This shows that $\{v_{1},y'\},\{v_{1},w'\}\in E$. With these facts we obtain the following $S$-cycles so that $v_{1}\notin C_{<ij}^{x'y',z'w'}$:
\begin{itemize}
\item If $\{v_{1},x'\}\in E$ then $\langle v_{1},x',y'\rangle$ is an induced $S$-triangle.
\item If $\{y',w'\}\in E$ then $\langle v_{1},y',w'\rangle$ is an induced $S$-triangle.
\item If $\{v_{1},x'\},\{y',w'\}\notin E$ then $\langle v_{1},y',x',w'\rangle$ is an induced $S$-square.
\end{itemize}
\item Let $\langle v_{1},v_{2},v_{3},b'\rangle$ be an $S$-square such that $v_{1},v_{2},v_{3}\in V_{<ij}\cup\{x',y',z',w',a'\}$.
If $v_{1}\in S$ or $v_{3}\in S$ then similar to the previous argument we can show that the $S$-vertex does not belong to $C_{<ij}^{x'y',z'w'}$.
Assume that $v_{2}\in S$. Since $x',y',z',w',a',b'\notin S$, $v_{2}\in V_{<ij}$.
Thus $v_2 <_t j <_t \{y,w\}$ and $v_2 <_b i <_b \{x,z\}$ which mean that $v_2 <_t \{y',w',b'\}$ and $v_2 <_b \{x',z',a'\}$.
Moreover $x'y'<_{\ell} z'w' <_{\ell} a'b'$ imply $x'<_{t}z'<_{t}a'$ and $y'<_{b}w'<_{b}b'$.
By the given $S$-square we have $\{v_1,v_3\} <_t \{v_2,b'\}$ and $\{v_2,b'\} <_b \{v_1,v_3\}$.
Since $\{v_2,b'\} \notin E$, we also have $v_2 <_t b'$ and $v_2 <_b b'$. Then the following hold:
\begin{itemize}
\item $\{v_{1},v_{3}\}<_{t}v_{2}<_{t}\{y',b'\}$ and
\item $\{v_{2},y'\}<_{b}b'<_{b}\{v_{1},v_{3}\}$.
\end{itemize}
Thus $\{v_{1},y'\},\{v_{3},y'\}\in E$. 
Now the following $S$-cycles show that $v_2 \notin C_{<ij}^{x'y',z'w'}$.
\begin{itemize}
\item If $\{v_{2},y'\}\in E$ then $\langle v_{1},v_{2},y'\rangle$ is an induced $S$-triangle.
\item If $\{v_{2},y'\}\notin E$ then $\langle v_{1},v_{2},v_{3},y'\rangle$ is an induced $S$-square.
\end{itemize}

\item Following the same lines as above we can show that if a subset of $V_{<ij}\cup\{x',y',z',w',a',b'\}$ that contains $a'$ induces an $S$-cycle of $G$ then its non-empty intersection with $V_{<ij}$ is not a subset of $C_{<ij}^{x'y',z'w'}$.
\end{itemize}
Therefore, if a subset of $V_{<ij}\cup\{x',y',z',w',a',b'\}$ that contains $a'$ or $b'$ induces an $S$-cycle then its non-empty intersection with $V_{<ij}$ is not a subset of $C_{<ij}^{x'y',z'w'}$.
By this fact it follows that $C_{ij}^{xy,zw}=C_{<ij}^{x'y',z'w'}\cup\{i,j\}$ as described in the third statement.
\end{proof}

It is important to notice that all described formulations are given recursively based on Lemmas~\ref{lem:sfvs:perm:i}--\ref{lem:sfvs:perm:ijxyzw}.
Now we are in position to state our claimed polynomial-time algorithm for the SFVS problem on permutation graphs.

\begin{theorem}\label{theo:permutation}
{\sc Subset Feedback Vertex Set} can be solved in $O(m^3)$ time on permutation graphs.
\end{theorem}
\begin{proof}
Let us describe such an algorithm.
Given the permutation diagram, that is the ordering $<_{t}$ and $<_{b}$ on $V$, we first compute all crossing pairs $ij$ of $\mathcal{X}$.
Observe that the number of such pairs is $n+m$.
For each crossing pair $ij$ we compute its predecessors $\left\{{\eqslantless}, {\leqslant}, {<}, {\ll}, {<} {\ll}\right\}$ according to the corresponding definition.
Note that such a simple application requires $O(n^2)$ time for every crossing pair $ij$, giving a total running time of $O(n^2m)$.
Next we scan all crossing pairs of $\mathcal{X}$ according to their ascending order with respect to $<_{r}$.
For every crossing pair $ij$ we compute $A_{ij}$ according to Lemmas~\ref{lem:sfvs:perm:i} and \ref{lem:sfvs:perm:ij}.
That is, for every crossing pair $xy$ of $V \setminus V_{ij}$ in descending order with respect to $<_{\ell}$ we compute $B_{ij}^{xy}$ according to Lemmas~\ref{lem:sfvs:perm:ix}, \ref{lem:sfvs:perm:ixy}, \ref{lem:sfvs:perm:ijx}, and \ref{lem:sfvs:perm:ijxy}.
By the recursive formulations of $B_{ij}^{xy}$, for every crossing pair $zw$ of $V \setminus V_{xy}$ in descending order with respect to $<_{\ell}$ we compute $C_{ij}^{xy,zw}$ according to Lemmas~\ref{lem:sfvs:perm:ixyz}, \ref{lem:sfvs:perm:ixyzw}, \ref{lem:sfvs:perm:ijxyz}, and \ref{lem:sfvs:perm:ijxyzw}.
In total all such computations require $O(m^3)$ time.
At the end the set $A_{\pi(n)n}$ is the maximum weighted $S$-forest for $G$ so that $V \setminus A_{\pi(n)n}$ is exactly the subset feedback vertex set of $G$.
\end{proof}

\section{Concluding remarks}
In accordance to Theorem~\ref{theo:circular} that uses the algorithm for interval graphs, we believe that our algorithm for permutation graphs can be suitably applied for the class of circular permutation graphs (the class of circular permutation graphs generalizes permutation graphs analogous to that of circular-arc graphs with respect to interval graphs \cite{graph:classes:brandstadt:1999,graph:classes:Go04}).
From the complexity point of view, since FVS is polynomial-time
solvable on the class of AT-free graphs \cite{KratschMT08},
a natural problem is to settle the complexity of SFVS on AT-free graphs.
Interestingly most problems that are hard on AT-free graphs are already hard on co-bipartite graphs (see for e.g., \cite{MAWSHANG1997135}).
Co-bipartite graphs are the complements of bipartite graphs and are unrelated to permutation graphs or interval graphs.
Here we show that SFVS admits a simple and efficient (polynomial-time) solution on co-bipartite
graphs and therefore excluding such an approach through a hardness result on co-bipartite graphs.

\begin{theorem}\label{theo:cobipartite}
The number of maximal $S$-forests of a co-bipartite graph is at most $22n^4$ and these can be enumerated in time $O(n^4)$.
\end{theorem}
\begin{proof}
Let $G=(V,E)$ be a co-bipartite graph and let $(A,B)$ be a partition of $V$ such that such that $G[A]$ and $G[B]$ are cliques.
We further partition $V$ as $(A_S,A_R,B_S,B_R)$
where $A_S = A \cap S$, $A_R = A \setminus S$,
$B_S = B \cap S$ and $B_R = B \setminus S$.
For a vertex $v$ of $G$ and a set $U \in \{A_S,B_S,A_R,B_R\}$ we denote by $N_{U}(v)$ the neighbours of $v$ in the set $U$, that is, $N_{U}(v) =\df N(v) \cap U$.
Moreover the {\it symmetric difference} of two sets $L$ and $R$ is the set $(L \setminus R) \cup (R \setminus L)$ and is denoted by $L \vartriangle R$.
Let $(X,Y,Z,W)$ be the partition of the vertex set of a maximal $S$-forest of $G$ such that $X \subseteq A_S$, $Y \subseteq A_R$, $Z \subseteq B_S$ and $W \subseteq B_R$.
It is clear that $|X| \leq 2$ and $|Z| \leq 2$.
Thus it is sufficient to consider the following cases with respect to $X$ and $Z$:
\begin{itemize}[itemsep=-0.1em,topsep=0pt]
\item Let $X=\emptyset$ and $Z=\emptyset$.
Then the maximal $S$-forest contains no vertex of $S$, so we can safely include all vertices of $V \setminus S$.
Thus the following set is a maximal $S$-forest of $G$:
\begin{enumerate}[topsep=0pt]
\item $\left(\emptyset,A_R,\emptyset,B_R\right)$.
\end{enumerate}

\item Let $X=\{a_S\}$ and $Z=\emptyset$. Observe that $|Y|\leq 1$, since $G[X \cup Y]$ is a clique.
If $Y=\emptyset$ then including at least two neighbours of $a_S$ that are contained in $B_R$ leads to an $S$-cycle.
Thus we can safely include all non-neighbours of $a_S$ and exactly one neighbour of $a_S$ contained in $B_R$ in the maximal $S$-forest.
If $Y=\{a_R\}$ then including a neighbour of $a_S$ and a neighbour of $a_R$ (may well be the same) that are contained in $B_R$ leads to an $S$-cycle.
If we do not include a neighbour of $a_S$ then we can safely include all other vertices of $B_{R}$.
However if we include a neighbour of $a_S$ that is non-adjacent to $a_R$ then we can safely include all other vertices that are non-adjacent to both.
Thus the following sets induce the corresponding maximal $S$-forests of $G$:
\begin{enumerate}[resume,topsep=0pt]
\item $\left(\{a_S\}, \emptyset, \emptyset, B_R\right)$, where $N_{B_R}(a_S)=\emptyset$;
\item $\left(\{a_S\}, \emptyset, \emptyset, \{b_R\} \cup \left(B_R \setminus N(a_S)\right)\right)$, where $b_R \in N_{B_{R}}(a_S)$;
\item $\left(\{a_S\}, \{a_R\}, \emptyset, B_R \setminus N(a_S)\right)$;
\item $\left(\{a_S\}, \{a_R\}, \emptyset, \{b_R\} \cup \left(B_R \setminus N(\{a_S,a_R\})\right)\right)$, where $b_R \in N_{B_{R}}(a_S) \setminus N_{B_{R}}(a_R)$.
\end{enumerate}

\item Let $X=\emptyset$ and $Z=\{b_S\}$. Completely symmetric arguments with the previous case imply that the following sets induce the corresponding maximal $S$-forests of $G$:
\begin{enumerate}[resume,topsep=0pt]
\item $\left(\emptyset, A_R, \{b_S\}, \emptyset\right)$, where $N_{A_R}(b_S)=\emptyset$;
\item $\left(\emptyset, \{a_R\} \cup \left(A_R \setminus N(b_S)\right), \{b_S\}, \emptyset\right)$, where $a_R \in N_{A_{R}}(b_S)$;
\item $\left(\emptyset, A_R \setminus N(b_S), \{b_S\}, \{b_R\}\right)$;
\item $\left(\emptyset, \{a_R\} \cup \left(A_R \setminus N(\{b_S,b_R\})\right), \{b_S\}, \{b_R\}\right)$, where $a_R \in N_{A_{R}}(b_S) \setminus N_{A_{R}}(b_R)$.
\end{enumerate}

\item Let $X=\{a_S\}$ and $Z=\{b_S\}$. Then both $|Y|\leq 1$ and $|W|\leq 1$.
Thus the following sets induce the maximal $S$-forest of $G$:
\begin{enumerate}[resume,topsep=0pt]
\item $\left(\{a_S\}, \emptyset, \{b_S\}, \emptyset\right)$, where $\{a_S,b_S\}\in E$ and $V \setminus S \subseteq N(a_S)\cap N(b_S)$;
\item $\left(\{a_S\}, \{a_R\}, \{b_S\}, \emptyset\right)$, where $G[\left\{a_S,a_R,b_S\right\}]$ is acyclic and $B_R \subseteq N(a_S) \cup N(a_R)$;
\item $\left(\{a_S\}, \emptyset, \{b_S\}, \{b_R\}\right)$, where $G[\left\{a_S,b_S,b_R\right\}]$ is acyclic and $A_R \subseteq N(b_S) \cup N(b_R)$;
\item $\left(\{a_S\}, \{a_R\}, \{b_S\}, \{b_R\}\right)$, where $G[\left\{a_S,a_R,b_S,b_R\right\}]$ is acyclic.
\end{enumerate}

\item Let $X=\{a_S,a'_S\}$ and $Z=\emptyset$. Then $|Y|=0$, since $G[X \cup Y]$ is a clique.
Adding a vertex of $B_R$ that is adjacent to both $a_S$ and $a'_S$ leads to an $S$-cycle.
If we add a vertex of $B_R$ that is adjacent to either $a_S$ or $a'_S$ then adding another such vertex leads to an $S$-cycle.
Thus we can safely include all other vertices that are non-adjacent to either $a_S$ or $a'_S$.
\begin{enumerate}[resume,topsep=0pt]
\item $\left(\{a_S,a'_S\}, \emptyset, \emptyset, B_R \setminus (N(a_S) \cap N(a'_S))\right)$, where $N_{B_R}(a_S) \vartriangle N_{B_R}(a'_S)=\emptyset$;
\item $\left(\{a_S,a'_S\}, \emptyset, \emptyset, \{b_R\} \cup \left(B_R\setminus N(\{a_S,a'_S\})\right)\right)$, where $b_R \in N_{B_R}(a_S) \vartriangle N_{B_R}(a'_S)$.
\end{enumerate}

\item Let $X=\emptyset$ and $Z=\{b_S,b'_S\}$. Completely symmetric arguments with the previous case imply that the following sets induce the corresponding maximal $S$-forest of $G$:
\begin{enumerate}[resume,topsep=0pt]
\item $\left(\emptyset, A_R \setminus (N(b_S) \cap N(b'_S)), \{b_S,b'_S\}, \emptyset\right)$, where $N_{A_R}(b_S) \vartriangle N_{A_R}(b'_S)=\emptyset$;
\item $\left(\emptyset, \{a_R\} \cup \left(A_R \setminus (N(b_S) \cup N(b'_S))\right), \{b_S,b'_S\}, \emptyset\right)$, where $a_R \in N_{A_R}(b_S) \vartriangle N_{A_R}(b'_S)$.
\end{enumerate}

\item Let $X=\{a_S,a'_S\}$ and $Z=\{b_S\}$. Then $|Y|=0$ and $|W|\leq 1$. Thus the following sets induce the maximal $S$-forest of $G$:
\begin{enumerate}[resume,topsep=0pt]
\item $\left(\{a_S,a'_S\}, \emptyset, \{b_S\}, \emptyset\right)$, where $G[\left\{a_S,a'_S,b_S\right\}]$ is acyclic;
\item $\left(\{a_S,a'_S\}, \emptyset, \{b_S\}, \{b_R\}\right)$, where $G[\left\{a_S,a'_S,b_S,b_R\right\}]$ is acyclic.
\end{enumerate}

\item Let $X=\{a_S\}$ and $Z=\{b_S,b'_S\}$.
Then similarly to the previous case we obtain the following:
\begin{enumerate}[resume,topsep=0pt]
\item $\left(\{a_S\}, \emptyset, \{b_S,b'_S\}, \emptyset\right)$, where $G[\left\{a_S,b_S,b'_S\right\}]$ is acyclic.
\item $\left(\{a_S\}, \{a_R\}, \{b_S,b'_S\}, \emptyset\right)$, where $G[\left\{a_S,b_S, b'_S, a\right\}]$ is acyclic.
\end{enumerate}

\item Let $X=\{a_S,a'_S\}$ and $Z=\{b_S,b'_S\}$. Then $|Y|=0$ and $|W|=0$ so that the following set induces such a maximal $S$-forest:
\begin{enumerate}[resume,topsep=0pt]
\item $\left(\{a_S,a'_S\}, \emptyset, \{b_S,b'_S\}, \emptyset\right)$, where $G[\left\{a_S,a'_S,b_S,b'_S\right\}]$ is acyclic.
\end{enumerate}
\end{itemize}
\medskip
\noindent Because $|X|, |Y|, |Z|, |W| \leq n$, each described maximal $S$-forest gives at most $n^4$ maximal $S$-forests.
Therefore in total there are at most $22n^4$ maximal $S$-forests that correspond to each particular case.
Taking into account that any maximal $S$-forest has at most $n$ vertices,
these arguments can be applied to obtain an enumeration algorithm that runs in time $O(n^4)$.
\end{proof}

We strongly believe that such an approach towards AT-free graphs should deal first with the complexity of the unweighted version of the SFVS.
Concerning graph classes related to bounded graph parameters, we would like to note that FVS is polynomial-time solvable on bounded clique-width graphs \cite{Bui-XuanSTV13}.
The computational complexity for SFVS on such graphs, or even on bounded linear clique-width graphs, has not been resolved yet and stimulates an intriguing open question.

Another interesting open question is concerned with problems related to {\it terminal-sets} such as the {\sc Multiway Cut} problem in which we want to disconnect a given set of terminals by removing vertices of minimum total weight.
As already mentioned in the Introduction, the {\sc Multiway Cut} problem reduces to the \textsc{Subset Feedback Vertex Set} problem by adding a vertex $s$ with $S = \{s\}$ that is adjacent to all terminals and whose weight is larger than the sum of the weights of all vertices in the original graph \cite{FominHKPV14}.
Notice that through such an approach in order to solve even the unweighted \textsc{Multiway Cut} problem one needs to solve the weighted \textsc{Subset Feedback Vertex Set} problem.
This actually implies that {\sc Multiway Cut} is polynomial-time solvable in permutation graphs and interval graphs by using our algorithms for the SFVS problem.
However polynomial-time algorithms for the {\sc Multiway Cut} problem were already known for permutation and interval graphs by a more general terminal-set problem \cite{GFKNU08,Papadopoulos12}.
Nevertheless it is still interesting to consider the computational complexity of the unweighted {\sc Multiway Cut} problem on subclasses of AT-free graphs such as co-comparability graphs.

\bibliography{subsetFVS_classes}{}
\bibliographystyle{plain}

\end{document}